\documentclass[reqno]{llncs}

\usepackage[latin1]{inputenc}
\usepackage[ngerman, english]{babel}
\usepackage{amsmath,amssymb,amsfonts}
\usepackage[T1]{fontenc}
\usepackage{bussproofs}
\usepackage{stmaryrd}
\usepackage{times}
\usepackage{verbatim}
\usepackage{mathrsfs}

\newcommand{\AC}{{\mathsf{AC}}}
\newcommand{\CWP}{{\mathsf{CWP}}}
\newcommand{\depth}{{\mathsf{depth}}}
\newcommand{\DET}{{\mathsf{DET}}}
\newcommand{\DLOGTIME}{{\mathsf{DLOGTIME}}}
\newcommand{\val}{{\mathsf{val}}}
\newcommand{\rhs}{{\mathsf{rhs}}}
\newcommand{\Id}{{\mathsf{Id}}}
\newcommand{\mdepth}{{\mathsf{mdepth}}}
\newcommand{\CL}{\mathsf{C}_=\mathsf{L}}
\newcommand{\CLOGCFL}{\mathsf{C}_=\mathsf{LogCFL}}
\newcommand{\NC}{{\mathsf{NC}}}
\newcommand{\TC}{{\mathsf{TC}}}
\newcommand{\UT}{{\mathsf{UT}}}

\begin{document}

\title{Evaluating Matrix Circuits}

\author{Daniel K\"onig \and Markus Lohrey}

\institute{Universit\"at Siegen, Germany}

\maketitle

\begin{abstract}
The circuit evaluation problem (also known as the compressed word problem) for
finitely generated linear groups is studied. The best upper bound for this problem is 
{\sf coRP}, which is shown by a reduction to polynomial identity testing. Conversely,
the compressed word problem for the linear group $\mathsf{SL}_3(\mathbb{Z})$ 
is equivalent to polynomial identity testing. In the paper, it is shown that the compressed
word problem for every finitely generated nilpotent group is in $\mathsf{DET} \subseteq \NC^2$.
Within the larger class of polycyclic groups we find examples where the compressed word problem
is at least as hard as polynomial identity testing for skew arithmetic circuits.
\end{abstract}

\section{Introduction}

The study of circuit evaluation problems has a long tradition in theoretical computer science
and is tightly connected to many aspects in computational complexity theory. One of the most
important circuit evaluation problems is \emph{polynomial identity testing}.
Here, the input is an arithmetic circuit, whose internal gates are labelled with either addition or multiplication
and its input gates are labelled with variables ($x_1, x_2, \ldots$) or constants ($-1,0,1$), and
it is asked whether the output gate evaluates to the zero polynomial (in this paper, we always
work in the polynomial ring over the coefficient ring $\mathbb{Z}$ or $\mathbb{Z}_p$ 
for a prime $p$). Based on the Schwartz-Zippel-DeMillo-Lipton Lemma,
Ibarra and Moran \cite{IbMo83} proved that polynomial identity testing over $\mathbb{Z}$ or $\mathbb{Z}_p$ belongs to the class {\sf coRP}
(the complements of problems in randomized polynomial time). Whether there is a deterministic
polynomial time algorithm for polynomial identity testing is an important problem.
In \cite{ImpWig97} it is shown that if there exists a language in $\mathsf{DTIME}(2^{\mathcal{O}(n)})$
that has circuit complexity $2^{\Omega(n)}$, then $\mathsf{P} = \mathsf{BPP}$ (and hence
$\mathsf{P} = \mathsf{RP} = \mathsf{coRP}$). There is also an implication that goes the other way round:
 Kabanets and Impagliazzo \cite{KabanetsI04} have shown that if polynomial identity testing
belongs to $\mathsf{P}$, then (i) there is a language in $\mathsf{NEXPTIME}$
 that does not have polynomial size circuits, or (ii)
the permanent is not computable by polynomial size arithmetic circuits.
Both conclusions represent major open problem in complexity theory. Hence, although
it is quite plausible that polynomial identity testing belongs to $\mathsf{P}$ (by \cite{ImpWig97}),
it will be probably very hard to prove (by \cite{KabanetsI04}).

Circuit evaluation problems can be also studied for other structures than polynomial rings, in particular
non-commutative structures. For finite monoids, the circuit evaluation problem has been studied in \cite{BeMcPeTh97},
where it was shown using Barrington's technique \cite{Bar89} that for every non-solvable finite monoid the circuit
evaluation problem is {\sf P}-complete, whereas for every solvable monoid, the circuit
evaluation problem belongs to the parallel complexity class $\DET \subseteq \NC^2$. Starting with 
\cite{Loh06siam} the circuit evaluation problem has been also studied for infinite finitely generated (f.g) monoids, in particular 
infinite  f.g. groups. In this context, the input gates of the circuit are labelled with generators of the monoid
and the internal gates compute the product of the two input gates.

In \cite{Loh06siam} and subsequent work, the circuit evaluation problem
is also called the {\em compressed word problem}. This is due to the fact that if one forgets the underlying
monoid structure of a multiplicative circuit, the circuit simply evaluates to a word over the monoid generators
labelling the input gates. This word can be of length exponential in the number of circuit gates. Hence,
the circuit can be seen as a compressed representation of the word it produces. In this context, 
circuits are also known as {\em straight-line programs} (SLPs) and are intensively studied in the area of 
algorithms for compressed data, see \cite{Loh12survey} for an overview.

Concerning the compressed word problem, polynomial time algorithms have been developed for
many important classes of groups, e.g., finite groups,  f.g. nilpotent groups,  f.g. free groups,
graph groups (also known as right-angled Artin groups or partially commutative groups), and 
virtually special groups. The latter contain all Coxeter groups, one-relator groups with torsion,
fully residually free groups, and fundamental groups of hyperbolic 3-manifolds; see \cite{Loh14} for details.
For the important class of f.g. linear groups, i.e., f.g. groups of matrices over a field, it was shown
in \cite{Loh14} that the compressed word problem can be reduced to polynomial identity testing (over $\mathbb{Z}$
or $\mathbb{Z}_p$, depending on the characteristic of the field) and hence belongs to $\mathsf{coRP}$.
Vice versa, in \cite{Loh14} it was shown that polynomial identity testing over $\mathbb{Z}$ can be reduced
to the compressed word problem for the linear group $\mathsf{SL}_3(\mathbb{Z})$. The proof is based on 
a construction of Ben-Or and Cleve  \cite{Ben-OrC92}.  This result indicates that derandomizing the compressed word problem
for a f.g. linear group will be in general very difficult.

In this paper, we further investigate the tight correspondence between commutative circuits over rings and 
non-commutative circuits over linear groups. In Section~ \ref{cwp-nilpotent} we study the complexity of 
the compressed word problem for f.g. nilpotent groups. For these groups, the compressed
word problem can be solved in polynomial time \cite{Loh14}. Here, we show that for every f.g.
nilpotent group the compressed word problem belongs to the parallel complexity class $\DET \subseteq \NC^2$,
which is the class of all problems that are $\NC^1$-reducible to the computation of the determinant of an 
integer matrix, see \cite{Coo85}. To the knowledge of the authors, 
f.g. nilpotent groups are the only examples of infinite groups for which the compressed word problem belongs to $\NC$.
Even for free groups, the compressed word problem is {\sf P}-complete \cite{Loh06siam}. The main step 
of our proof for f.g. nilpotent groups is to show that 
for a torsion-free f.g. nilpotent group $G$ the compressed word problem belongs to the logspace
counting class  $\CL$ (and is in fact $\CL$-complete if $G$ is nontrivial). To show this, we use the well-known
fact that a f.g.  torsion-free nilpotent group can be embedded into the group $\UT_d(\mathbb{Z})$ of $d$-dimensional
unitriangular matrices over $\mathbb{Z}$ for some fixed $d$. Then, the compressed word problem for 
$\UT_d(\mathbb{Z})$ is reduced to the question whether two additive circuits over the natural numbers
evaluate to the same number, which is $\CL$-complete. Let us mention that there are several $\CL$-complete
problems related to linear algebra \cite{AllenderBO99}.

We also study  the compressed word problem for the matrix group $\UT_d(\mathbb{Z})$ for the case that the 
dimension $d$ is not fixed, i.e., part of the input (Section~\ref{sec-uniform}). In this case, the compressed word problem turns out to be complete
for the counting class  $\CLOGCFL$, which is the $\mathsf{LogCFL}$-analogue of  $\CL$.

Finally, in Section~\ref{sec-polycyclic} we move from nilpotent groups to polycyclic groups. These are solvable groups, where every
subgroup is finitely generated. By results of Maltsev, Auslander, and Swan these are exactly the solvable subgroups
of $\mathsf{GL}_d(\mathbb{Z})$ for some $d$. We prove that polynomial identity testing for skew arithmetic circuits reduces to 
the compressed word problem for a specific 2-generator polycyclic group of Hirsch length 3. 
A skew arithmetic circuit is an arithmetic circuit (as defined in the first paragraph of the introduction) such that for every
multiplication gate, one of its input gates is an input gate of the circuit, i.e., a variable or a constant. 
These circuits exactly correspond to algebraic branching programs.
Even for skew arithmetic circuits, no polynomial time algorithm is currently known (although the problem belongs to 
$\mathsf{coRNC}$).

\section{arithmetic circuits}

We use the standard notion of (division-free) arithmetic 
circuits. Let us fix a set $X = \{ x_1, x_2, \ldots \}$ of 
variables.
An \emph{arithmetic circuit} is a triple $C = (V,S,\rhs)$ with the following properties:
\begin{itemize}
\item $V$ is a finite set of \emph{gates}.
\item $S \in V$ is the {\em output gate}.
\item For every gate $A$, $\rhs(A)$ (the \emph{right-hand side of $A$ }) is either a variable from $X$, one of the constants $-1$, $0$, $1$, or an expression of 
the form $B+C$ or $B \cdot C$, where $B$ and $C$ are gates.
\item  There is a linear order $<$ on $V$ such that $B < A$ whenever $B$ occurs in $\rhs(A)$.
\end{itemize}
A gate $A$, where $\rhs(A)$ has the form $B+C$ (resp., $B \cdot C$)  
is called an addition gate (resp., multiplication gate). A gate that is labelled with a variable
or a constant is an \emph{input gate}.

Fix a ring $(R,+,\cdot)$ (which will be $(\mathbb{Z},+,\cdot)$ in most cases) 
and assume that $\mathcal{C} = (V,S,\rhs)$ is an arithmetic circuit 
in which the variables $x_1, \ldots, x_n$ occur. Then we can evaluate every gate $A \in V$ to a polynomial
$\val_{\mathcal{C}}(A) \in R[x_1,\ldots, x_n]$ in the obvious way (here,``$\val$'' stands for ``value'').
Moreover let $\val(\mathcal{C}) = \val_{\mathcal{C}}(S)$ be the polynomial to which $\mathcal{C}$ evaluates.
Two arithmetic circuits $\mathcal{C}_1$ and $\mathcal{C}_2$ 
are equivalent if they evaluate to the same polynomial.

Fix an arithmetic circuit $\mathcal{C} = (V,S,\rhs)$. 
We can view $\mathcal{C}$ as a directed acyclic graph (dag), where every node is labelled with 
a variable or a constant or an operator $+$, $\cdot$. If $\rhs(A) = B \circ C$ (for $\circ$ one of the operators),
then there is an edge from $B$ to $A$ and $C$ to $A$. 
The \emph{depth} $\mathsf{depth}(A)$ (resp., \emph{multiplication depth} $\mathsf{mdepth}(A)$)
of the gate $A$ is the maximal number of gates (resp., multiplication gates) along a path from 
an input gate to $A$. So, input gates have depth one and multiplication depth zero.
The \emph{depth} (resp., \emph{multiplication depth})  of
$\mathcal{C}$ is $\depth(\mathcal{C}) = \depth(S)$ (resp., $\mdepth(\mathcal{C}) = \mdepth(S)$).
The \emph{formal degree} $\mathsf{deg}(A)$ of a gate $A$ is $1$ if $A$ is an input gate, 
$\max\{ \mathsf{deg}(B),\mathsf{deg}(C) \}$ if $\rhs(A) = B+C$, and  
$\mathsf{deg}(B) + \mathsf{deg}(C)$ if $\rhs(A) = B \cdot C$. The formal degree 
of $\mathcal{C}$ is $\mathsf{deg}(\mathcal{C}) = \deg(S)$.
A \emph{positive circuit} is an arithmetic circuit without input gates labelled by the constant $-1$.
An \emph{addition circuit} is  a positive circuit without multiplication gates. A \emph{variable-free circuit} is a circuit without variables. 
It evaluates to an element of the underlying ring. 
A \emph{skew circuit} is an arithmetic circuit such that for every multiplication gate $A$ with $\rhs(A) = B \cdot C$,
one of the gates $B,C$ is an input gate.

In the rest of the paper we will also allow more complicated expressions in 
right-hand sides for gates. For instance, we may have a gate with
$\rhs(A) = (B+C) \cdot (D+E)$.  When writing down such a right-hand side, we implicitly assume
that there are additional gates in the circuit, with (in our example) right hand sides $B+C$ and 
$D+E$, respectively. The following lemma is folklore. We give a proof for completeness.

\begin{lemma} \label{lemma-eliminate-sub}
Given an arithmetic circuit $\mathcal{C}$ one can compute in logarithmic space two
positive circuits $\mathcal{C}_1$ and $\mathcal{C}_2$ such that $\val(\mathcal{C}) = 
\val(\mathcal{C}_1)-\val(\mathcal{C}_2)$ for every ring. Moreover, for $i \in \{1,2\}$
we have $\deg(\mathcal{C}_i) \leq \deg(\mathcal{C})$,  $\depth(\mathcal{C}_i) \leq 2 \cdot \depth(\mathcal{C})$,
and $\mdepth(\mathcal{C}_i) \leq \mdepth(\mathcal{C})$.
\end{lemma}

\begin{proof}
Let $\mathcal{C}=(V,S,\rhs)$ be an arithmetic circuit. We define the positive circuits $\mathcal{C}_1=(V',S_1,\rhs')$ and $\mathcal{C}_2=(V',S_2,\rhs')$ 
as follows:
\begin{gather*} 
 V'= \left\{ A_{i} \mid A \in V, i \in \left\{ 1,2 \right\} \right\}, \\
\rhs'(A_{i})= B_{i} + C_{i}   \text { if } \rhs(A)=B+C \text{ for } i \in \{1,2 \}, \\
\rhs'(A_1)=B_{1}C_{1}+B_{2}C_{2} \text { if } \rhs(A)=B \cdot C, \\
\rhs'(A_2)=B_{1}C_{2}+B_{2}C_{1}  \text { if } \rhs(A)=B \cdot C, \\ 
\rhs'(A_{1})= \rhs(A) \text{ if } \rhs(A) \in \{0,1\} \cup X, \\
\rhs'(A_{2})=0 \text{ if } \rhs(A) \in \{0,1\} \cup X , \\
\rhs'(A_{1})= 0 \text{ if } \rhs(A)=-1, \\
\rhs'(A_{2})=1 \text{ if } \rhs(A)= -1 .
\end{gather*}
Now we show by induction that for every gate $A \in V$ we have $\val(A)=\val(A_1)-\val(A_2)$:
The case that $A$ is an input gate is trivial. 
Now let $A$ be an addition gate with $\rhs(A)=B+C$ such that the statement is true for $B$ and $C$. Then 
\begin{align*}
\val(A)
 & = \val(B) + \val(C)  \\
&= \val(B_{1}) - \val(B_{2})+ \val(C_{1}) - \val(C_{2}) \\
&= (\val(B_{1}) + \val(C_{1}))   - (\val(B_{2}) + \val(C_{2}))  \\
&=\val(A_{1})-\val(A_{2}) 
\end{align*}
Finally, let $A$ be a multiplication gate with $\rhs(A)=B \cdot C$. We get 
\begin{align*}
 \val(A)&= \val(B)\val(C)\\
 &=(\val(B_{1})-\val(B_{2}))(\val(C_{1})-\val(C_{2}))\\
 &=\val(B_{1})\val(C_{1}) + \val(B_{2})\val(C_{2})
 -\val(B_{1})\val(C_{2})- \val(B_{2})\val(C_{1})\\
 &= \val(A_1)- \val(A_2) .
\end{align*}
So the claim holds.
The construction of $\mathcal{C}_1$ and $\mathcal{C}_2$ 
can be done in logarithmic space. By induction, it can be shown that for every gate $A$ of $\mathcal{C}$ and every $i \in \{1,2\}$,
one has $\deg(A_i) = \deg(A)$,  $\depth(A_i) \leq 2\cdot \depth(A)$,
and $\mdepth(A_i) = \mdepth(A)$.
\qed
\end{proof}
\emph{Polynomial identity testing} for a ring $R$ is the following computational problem: Given an arithmetic
circuit $\mathcal{C}$ (with variables $x_1, \ldots, x_n$), does $\val(\mathcal{C}) = 0$ hold, i.e., does $\mathcal{C}$ evaluate
to the zero-polynomial in $R[x_1, \ldots, x_n]$? It is an outstanding open problem in algebraic complexity theory, whether polynomial identity testing for $\mathbb{Z}$
can be solved in polynomial time.

\section{Complexity classes}
\label{sec-complexity}

The counting class $\#\mathsf{L}$ consists of all functions $f : \Sigma^* \to \mathbb{N}$ for which there is a logarithmic space bounded nondeterministic
Turing machine $M$ such that for every $w \in \Sigma^*$, $f(w)$ is the number of accepting computation paths of $M$ on input $x$.
The class $\CL$ contains all languages $A$ for which there are two functions $f_1, f_2 \in \#\mathsf{L}$ such that for every $w \in \Sigma^*$,
$w \in A$ if and only if $f_1(w) = f_2(w)$.
The class $\CL$ is closed under logspace many-one reductions.
The canonical $\CL$-complete problem is the following: The input consists of two dags $G_1$ and $G_2$ and vertices $s_1, t_1$ (in $G_1$) and $s_2, t_2$ (in $G_2$),
and it is asked whether the number of different paths from $s_1$ to $t_1$ in $G_1$ is equal to the number of different paths from $s_2$ to $t_2$ in $G_2$.
This problem is easily seen to be equivalent to the following problem: Given two variable-free addition circuits $\mathcal{C}_1$ and $\mathcal{C}_2$, does $\val(\mathcal{C}_1) = 
\val(\mathcal{C}_2)$ hold? Several $\CL$-complete problem is the question whether the determinant of a given integer matrix is zero \cite{Toda91countingproblems,Vinay91}.

We use standard definitions concerning circuit complexity, see e.g. \cite{Vol99} for more details. 
In particular we will consider the class $\TC^0$ of all problems
that can be solved by a polynomial size circuit family of constant depth that uses NOT-gates and unbounded fan-in AND-gates, OR-gates, and majority-gates.
For {\sf DLOGTIME}-uniform $\TC^0$ it is required in addition that for binary coded gate numbers $u$ and $v$, one can (i) compute the type of gate $u$ in time
$O(|u|)$ and (ii) check in  time $O(|u|+|v|)$ whether $u$ is an input gate for $v$.
Note that the circuit for inputs of length $n$ has at most $p(n)$ gates for a polynomial $p(n)$. Hence, the binary codings $u$ and $v$ have
length $O(\log n)$, i.e., the above computations can be done in $\mathsf{DTIME}(\log n)$. This is the reason for using the term 
``{\sf DLOGTIME}-uniform''.  If majority gates are not allowed, we obtain the class ({\sf DLOGTIME}-uniform) $\AC^0$. 
The class ({\sf DLOGTIME}-uniform) $\NC^1$ is defined by ({\sf DLOGTIME}-uniform) polynomial size circuit families
of logarithmic depth that use NOT-gates and  fan-in-2 AND-gates and OR-gates.
A language $A$ is $\AC^0$-reducible
to languages $B_1, \ldots, B_k$ if $A$ can be solved with a 
{\sf DLOGTIME}-uniform polynomial size circuit family of constant depth that uses NOT-gates and unbounded fan-in AND-gates, OR-gates, and 
$B_i$-gates ($1 \leq i \leq k$). Here, a $B_i$-gate (it is also called an oracle gate) receives an ordered tuple of inputs $x_1, x_2, \ldots, x_n$ and outputs $1$ if and only if 
$x_1  x_2 \cdots x_n \in B_i$. Sometimes, also the term ``uniform constant depth reducibility'' is used for this type of reductions.
In the same way, the weaker $\NC^1$-reducibility can be defined. Here, one counts the depth of a $B_i$-gate with inputs
$x_1, x_2, \ldots, x_n$ as $\log n$.  The class $\DET$ contains all problems that are $\NC^1$-reducible to the computation 
of the determinant of an integer matrix, see \cite{Coo85}.
It is known that $\CL \subseteq \DET \subseteq \NC^2$, see e.g. \cite[Section 4]{AlJe93}.

An \emph{NAuxPDA} is a nondeterministic Turing machine with an additional pushdown store.
The class $\mathsf{LogCFL} \subseteq \NC^2$ is the class of all languages that can be accepted by a polynomial time bounded NAuxPDA whose 
work tape is logarithmically bounded (but the pushdown store is unbounded). If we assign to the input the number of accepting
computation paths of such an NAuxPDA, we obtain the counting class $\#\mathsf{LogCFL}$.
In \cite{Vinay91} it is shown that $\#\mathsf{LogCFL}$ is the class of all 
functions $f : \{0,1\}^* \to \mathbb{N}$ (a non-binary input alphabet $\Sigma$ has to be encoded into $\{0,1\}^*$) 
for which there exists a logspace-uniform family $(\mathcal{C}_n)_{n \geq 1}$ of positive arithmetic circuits such that 
$\mathcal{C}_n$ computes the mapping $f$ restricted to $\{0,1\}^n$ and there is a polynomial $p(n)$ such that 
the formal degree of $\mathcal{C}_n$ is bounded by $p(n)$.
The class $\CLOGCFL$ contains all languages $A$ for which there are two functions $f_1, f_2 \in \#\mathsf{LogCFL}$ such that for every $w \in \Sigma^*$,
$w \in A$ if and only if $f_1(w) = f_2(w)$.
We need the following simple lemma, whose proof is based on folklore ideas:

\begin{lemma} \label{lemma-circuit-NAuxPDA}
There is an NAuxPDA $\mathcal{P}$ that gets as input a positive variable-free arithmetic circuit $\mathcal{C} = (V,S,\rhs)$ and such that
the number of accepting computations of $\mathcal{P}$ on input $\mathcal{C}$ is
$\val(\mathcal{C})$. 
Moreover, the running time is bounded polynomially in $\mathsf{depth}(\mathcal{C}) \cdot \mathsf{deg}(\mathcal{C})$.
\end{lemma}

\begin{proof}
The NAuxPDA $\mathcal{P}$ stores a sequence of gates on its pushdown (every gate can be encoded using $\log(|V|)$ many bits). 
In the first step it pushes the output gate $S$ on the initially empty pushdown. If $A$ is on top of the pushdown and 
$\rhs(A) = B+C$, then $\mathcal{P}$ replaces $A$ on the pushdown by $B$ or $C$, where the choice is made nondeterministically.
If $\rhs(A) = B \cdot C$, then $\mathcal{P}$ replaces $A$ on the pushdown by $BC$.
If $\rhs(A) = 0$, then $\mathcal{P}$ terminates and rejects. 
Finally, if $\rhs(A) = 1$, then $\mathcal{P}$ pops $A$ from the pushdown. If thereby the pushdown becomes empty then
$\mathcal{P}$  terminates and accepts.
In addition to its pushdown, $\mathcal{P}$ only needs a logspace bounded work tape to store a single gate. 
Moreover, if we start $\mathcal{P}$ with only the gate $A$ on the pushdown, then (i) the number of accepting
computation paths from that configuration is exactly $\val_{\mathcal{C}}(A)$ and (ii) 
the number of pushdown operations along a computation path is bounded by $\depth(A) \cdot \deg(A)$.
Both statements follow easily by induction.
\qed
\end{proof}

\section{Matrices and groups} \label{sec-matrices}

Let $A$ be a square matrix of dimension $d$ over some commutative ring  $R$. With $A[i,j]$ we denote the entry of $A$ in row $i$ and column $j$.
The matrix $A$ is called \emph{triangular} if $A[i,j] = 0$ whenever $i > j$, i.e., all entries below the main diagonal are $0$.
A \emph{unitriangular matrix} is a triangular matrix $A$ such that $A[i,i] = 1$ for all $1 \leq i \leq d$, i.e., all entries on the main 
diagonal are $1$. We denote the set of unitriangular matrices of dimension $d$ over the ring $R$ by $\UT_d(R)$. It is well known that for every 
commutative ring $R$, the set  $\UT_d(R)$ is a group (with respect to matrix multiplication).

Let $1 \leq i  < j \leq d$. With $T_{i,j}$ we denote the matrix from $\UT_d(R)$
such that $T_{i,j}[i,j] = 1$ and $T_{i,j}[k,l] = 0$ for all $k,l$ with $1 \leq k  < l \leq d$ and 
$(k,l) \neq (i,j)$. The notation $T_{i,j}$ does not specify the dimension $d$ of the matrix,
but the dimension will be always clear from the context.
The group $\UT_d(\mathbb{Z})$ is generated by the finite set 
$\Gamma_d = \{ T_{i,i+1} \mid 1 \leq i < d\}$, see e.g. \cite{BiSa01}. 

As usual we denote with $[x,y] = x^{-1} y^{-1} x y$ the commutator of $x$ and $y$.
We will make use of the following lemma, which shows how to encode multiplication
with unitriangular matrices. See \cite{Lohrey25unitri} for a proof.

\begin{lemma} \label{lemma-product}
For all $a,b \in \mathbb{Z}$ and $1 \leq i < j < k \leq d$ we have $[T_{i,j}^a, T_{j,k}^b] =  T_{i,k}^{ab}$.
\end{lemma}
In this paper we are concerned with certain subclasses of \emph{linear groups}. 
A group is linear if it is isomorphic to a subgroup of $\mathsf{GL}_d(F)$ (the group
of all invertible $(d \times d)$-matrices over the field $F$) for some field $F$.

A ($n$-step) solvable group $G$ is a group $G$, which has a 
a subnormal series $G = G_n \rhd G_{n-1} \rhd G_{n-2} \rhd \cdots \rhd G_1 \rhd G_0 = 1$
(i.e., $G_i$ is a normal subgroup of $G_{i+1}$ for all $0 \leq i \leq n-1$) such that 
every quotient $G_{i+1}/G_i$ is abelian ($0 \leq i \leq n-1$).
If every quotient $G_{i+1}/G_i$ is cyclic, then $G$ is called {\em polycyclic}.
The number of $0 \leq i \leq n-1$ such that $G_{i+1}/G_i \cong \mathbb{Z}$ is called the \emph{Hirsch length} of $G$; it does 
not depend on the chosen subnormal series.
If $G_{i+1}/G_i \cong \mathbb{Z}$ for all $0 \leq i \leq n-1$ then $G$ is called \emph{strongly polycyclic}. 
A group is polycyclic if and only if it is solvable and every subgroup is finitely generated. 
Polycyclic groups are linear. More precisely,
Auslander and Swan \cite{Aus67,Swa67}  proved that the polycyclic groups are exactly the 
solvable groups of integer matrices.

For a group $G$ its \emph{lower central series} is the series $G = G_1 \rhd G_2 \rhd G_3 \rhd \cdots$ of subgroups,
where $G_{i+1}  = [G_i,G]$, which is the subgroup generated by all commutators $[g,h]$ with $g \in G_i$ and $h \in G$.
Indeed, $G_{i+1}$ is a normal subgroup of $G_i$. The group $G$ is \emph{nilpotent}, if its lower central series terminates
after finitely many steps in the trivial group $1$. Every f.g. nilpotent group is polycyclic.
We need the following results about nilpotent and solvable groups:

\begin{theorem}[Chapter 5 in \cite{Rot95}]
 \label{thm-nilpotent-torsion-free-subgroup}
Every subgroup and every quotient of a solvable (resp., nilpotent) group $G$ is solvable (resp., nilpotent) again.
\end{theorem}

\begin{theorem}[Theorem 17.2.2 in \cite{KaMe79}] \label{thm-nilpotent-subgroup}
Every f.g. nilpotent group $G$ has a torsion-free normal subgroup $H$ of finite index (which is also f.g. nilpotent).
\end{theorem}

\begin{theorem}[Theorem 17.2.5 in \cite{KaMe79}] \label{thm-embed-nilpotent}
For every torsion-free f.g nilpotent group $G$ there exists $d \geq 1$ such that 
$G$ can be embedded into $\UT_d(\mathbb{Z})$.
\end{theorem}
A group $G$ is called metabelian if the commutator subgroup $[G,G]$ is abelian. In other words,
the metabelian groups are the 2-step solvable groups.
Even if $G$ is f.g. metabelian, this does not 
imply that $G$ is polycyclic, since $[G,G]$ is not necessarily finitely generated.

Let $G$ be a f.g. group and let $G$ be finitely generated as a group by $\Sigma$.
Then, as a monoid $G$ is finitely generated by $\Sigma \cup \Sigma^{-1}$ (where $\Sigma^{-1} = \{ a^{-1} \mid a \in \Sigma \}$
is a disjoint copy of $\Sigma$ and $a^{-1}$ stands for the inverse of the generator $a \in \Sigma$). 
Recall that the {\em word problem} for $G$ 
is the following computational problem: Given a string $w \in (\Sigma \cup \Sigma^{-1})^*$, does 
$w$ evaluate to the identity of $G$. Kharlampovich proved that there exist finitely presented 3-step solvable groups with an undecidable
word problem. On the other hand, for every f.g. linear group the word problem can be solved in deterministic
logarithmic space by results of Lipton and Zalcstein \cite{LiZa77} and Simon \cite{Sim79}. This applies in particular
to polycyclic groups. Robinson proved in his thesis that the word problem for a polycyclic group belongs to 
$\TC^0$ \cite{Rob93}, but his circuits are not uniform. Waack considered in \cite{Waa91} arbitrary f.g. solvable linear groups
(which include the polycyclic groups) and proved that their word problems belong to logspace-uniform $\NC^1$. 
In the appendix we combine Waack's technique with the famous division breakthrough results by 
Hesse, Allender, and Barrington \cite{HeAlBa02} to show that for every f.g. solvable linear group the word problem belongs
to $\DLOGTIME$-uniform $\TC^0$ (we decided to move this result to the appendix, sind the classical word problem for groups
is not the main focus of this paper).

\section{Straight-line programs and the compressed word problem}

A straight-line program (briefly, SLP) is basically a multiplicative circuit over a monoid. We define an SLP
over the finite alphabet $\Sigma$
as a triple $\mathcal{G}  = (V,S,\rhs)$, where $V$ is a finite set of variables (or gates),
$S \in V$ is the start variable (or output gate), and $\rhs$ maps every variable to a right-hand side $\rhs(A)$, which
is either a symbol $a \in \Sigma$, or of the form $BC$, where $B, C \in V$. As for arithmetic circuits
we require that there is a linear order $<$ on $V$ such that $B < A$, whenever $B$ occurs in $\rhs(A)$.
The terminology ``(start) variable'' (instead of ``(output) gate'') comes from the fact that an SLP is quite often
defined as a context-free grammar that produces a single string over $\Sigma$. This string is defined in the obvious way
by iteratively replacing variables by the corresponding right-hand sides, starting with the start variable.
We denote this string with $\val(\mathcal{G})$. The unique string over $\Sigma$, derived from the variable $A \in V$,
is denoted with $\val_{\mathcal{G}}(A)$. We will also allow more general right-hand sides from $(V \cup \Sigma)^*$,
but by introducing new variables we can always obtain an equivalent SLP in the above form.

If we have a monoid $M$, which is finitely generated by the set $\Sigma$, then there exists a canonical monoid
homomorphism $h : \Sigma^* \to M$.  Then, an SLP $\mathcal{G}$ over the alphabet $\Sigma$ can be evaluated
over the monoid $M$, which yields the monoid element  $h(\val(\mathcal{G}))$. In this paper,  we are only interested
in the case that the monoid $M$ is a f.g. group $G$. Let $G$ be finitely generated as a group by $\Sigma$. An SLP over the alphabet
$\Sigma \cup \Sigma^{-1}$ is also called an SLP over the group $G$. In this case, we will quite often identify the string $\val(\mathcal{G})
\in (\Sigma \cup \Sigma^{-1})^*$ with the group element $g \in G$ to which it evaluates. We will briefly write ``$\val(\mathcal{G}) = g$ in $G$''
in this situation. 

The main computational problem we are interested in is the \emph{compressed word problem} for a f.g. group $G$ (with 
a finite generating set $\Sigma$), briefly $\CWP(G)$. The input for this problem is an SLP $\mathcal{G}$ over the alphabet $\Sigma \cup \Sigma^{-1}$, and it 
is asked whether $\val(\mathcal{G}) = 1$ in $G$ (where of course $1$ denotes the group identity). The term ``compressed word problem''
comes from the fact that this problem can be seen as a succinct version of the classical word problem for $G$, where the  input
is an explicitly given string $w \in (\Sigma \cup \Sigma^{-1})^*$ instead of an SLP-compressed string.

The compressed word problem is related to the classical word problem. For instance, the classical word problem for 
a f.g. subgroup of the automorphism group of a group $G$ can be reduced to the compressed word problem for $G$,
and similar results are known for certain group extensions, see \cite{Loh14} for more details. Groups, for which the compressed
word problem can be solved in polynomial time are  \cite{Loh14}: finite groups,  f.g. nilpotent groups, f.g. free groups,
graph groups (also known as right-angled Artin groups or partially commutative groups), and 
virtually special groups, which are groups that have a finite index subgroup that embeds into a graph group.
The latter groups form a rather large class that include for instance Coxeter groups, one-relator groups with torsion,
residually free groups, and fundamental groups of hyperbolic 3-manifolds.
In \cite{BeMcPeTh97} the parallel complexity of the compressed word problem (there, called the circuit evaluation problem) 
for finite groups was studied, and the following result was shown:

\begin{theorem}[\cite{BeMcPeTh97}] \label{thm-CWP-finite-groups}
Let $G$ be a finite group. If $G$ is solvable, then $\CWP(G)$ belongs to the class $\DET \subseteq \NC^2$.
If $G$ is not solvable, then $\CWP(G)$ is $\mathsf{P}$-complete.
\end{theorem}
The following two results are proven in \cite{Loh14}. Recall that $\mathsf{RP}$ is the set of all problems
$A$ for which there exists a polynomial time bounded randomized Turing machine $R$ such that: (i) if 
$x \in A$ then $R$ accepts $x$ with probability at least $1/2$, and (ii) if $x \not\in A$ then $R$ accepts
$x$ with probability $0$. The class $\mathsf{coRP}$ is the class of all complements of problems from $\mathsf{RP}$.

\begin{theorem}[Theorem 4.15 in \cite{Loh14}] \label{thm-CWP-linear-groups}
For every f.g. linear group the compressed word problem belongs to the class $\mathsf{coRP}$.
\end{theorem}
This result is shown by reducing the compressed word problem for a f.g. linear group to polynomial identity testing 
for the ring $\mathbb{Z}$. Also a kind of converse of Theorem~ \ref{thm-CWP-linear-groups} is shown in \cite{Loh14}:

\begin{theorem}[Theorem 4.16 in \cite{Loh14}] \label{thm-CWP-SL3}
The problem $\CWP(\mathsf{SL}_3(\mathbb{Z}))$ and polynomial identity testing
for the ring $\mathbb{Z}$ are polynomial time reducible to each other.
\end{theorem}
This result is shown by using the construction of Ben-Or and Cleve \cite{Ben-OrC92} for simulating arithmetic circuits by matrix products.

\section{The compressed word problem for finitely generated nilpotent groups} \label{cwp-nilpotent}

The main result of this section is:

\begin{theorem} \label{thm-nilpotent-main}
Let $G \neq 1$ be a f.g. torsion-free nilpotent group. Then $\CWP(G)$ is complete for the class $\CL$.
\end{theorem}
For the lower bound let $G$ be a  non-trivial  f.g. torsion-free nilpotent group. Since $G \neq 1$, $G$ contains $\mathbb{Z}$.
Hence, it suffices to prove the following:

\begin{lemma}
$\CWP(\mathbb{Z})$ is hard for $\CL$.
\end{lemma}

\begin{proof}
Clearly, an SLP $\mathcal{G}$ over the generator $1$ of $\mathbb{Z}$ and its inverse $-1$ is nothing else than a variable-free arithmetic circuit $\mathcal{C}$ without 
multiplication gates. Using Lemma~\ref{lemma-eliminate-sub} we can construct in logspace
two addition circuits $\mathcal{C}_1$ and $\mathcal{C}_2$ such that $\val(\mathcal{C}) = 0$ if and only if 
$\val(\mathcal{C}_1) = \val(\mathcal{C}_2)$. Checking the latter identity is complete for $\CL$
as remarked in Section~\ref{sec-complexity}.
\qed
\end{proof}
For the upper bound in Theorem~\ref{thm-nilpotent-main}, we use the fact that every torsion-free f.g. nilpotent group can be embedded into the group $\UT_d(\mathbb{Z})$ for some $d \geq 1$
(Theorem~\ref{thm-embed-nilpotent}). Hence, it suffices to show the following result:

\begin{lemma} \label{lemma-UT-in-CL}
For every $d \geq 1$, $\CWP(\UT_d(\mathbb{Z}))$ belongs to $\CL$.
\end{lemma}
For the rest of this section let us fix a number $d \geq 1$ and consider the unitriangluar matrix group $\UT_d(\mathbb{Z})$.
Consider an SLP $\mathcal{G} = (V,S,\rhs)$ over the alphabet $\Gamma_d \cup \Gamma_d^{-1}$, where $\Gamma_d$
is the finite generating set of $\UT_d(\mathbb{Z})$ from Section~\ref{sec-matrices}.
Note that for every variable $A \in V$, $\val_{\mathcal{G}}(A)$ is a word over the alphabet $\Gamma_d \cup \Gamma_d^{-1}$. We identify in the following this word
with the matrix to which it evaluates. Thus, $\val_{\mathcal{G}}(A) \in \UT_d(\mathbb{Z})$.

Assume we have given  an arithmetic circuit $\mathcal{C}$. A partition $\biguplus_{i=1}^m V_i$ of the set of all multiplication gates of $\mathcal{C}$
is called \emph{structure-preserving} if for all multiplication gates $u,v$ of $\mathcal{C}$  the following holds: If there is a non-empty path from $u$ to $v$ in (the dag corresponding to)
$\mathcal{C}$ then there exist $1 \leq i < j \leq d$ such that $u \in V_i$ and $v \in V_j$.
In a first step, we transform our SLP $\mathcal{G}$ in logarithmic space into a 
variable-free arithmetic circuit $\mathcal{C}$ of multiplication depth at most $d$ such that $\mathcal{G}$ evaluates to the identity
matrix if and only if $\mathcal{C}$ evaluates to $0$. Moreover, we also compute a structure-preserving partition of the multiplication gates of 
$\mathcal{C}$. This partition will be needed for the further computations.
 The degree bound in the following lemma will be only needed in Section~\ref{sec-uniform}.
 
\begin{lemma} \label{lemma-SLP->circuit}
From the SLP $\mathcal{G}= (V,S,\rhs)$ we can compute in logspace a variable-free arithmetic circuit $\mathcal{C}$ 
with $\mdepth(\mathcal{C}) \leq d$ and $\deg(\mathcal{C}) \leq 2(d-1)$,
such that $\val(\mathcal{G}) = \Id_d$ if and only if $\val(\mathcal{C}) = 0$. In addition we can compute in logspace a structure-preserving partition
$\biguplus_{i=1}^d V_i$ of the set of all multiplication gates of $\mathcal{C}$.
\end{lemma}

\begin{proof}
The set of gates of the circuit $\mathcal{C}$ is 
$$
W = \{ A_{i,j} \mid A \in V, 1 \leq i < j \leq d \} \cup \{ T \},
$$
where $T$ is the output gate.
The idea is simple: Gate $A_{i,j}$ will evaluate to the matrix entry $\val_{\mathcal{G}}(A)[i,j]$. To achieve this, we define 
the right-hand side mapping of the circuit $\mathcal{G}$ (which we denote again with $\rhs$) as follows:
\begin{equation*}
\rhs(A_{i,j})= \begin{cases} M[i,j]  & \text{ if } \rhs(A)=M \in \Gamma_d \cup \Gamma_d^{-1}  \\
  B_{i,j} + C_{i,j} + \sum_{i < k < j} B_{i,k} \cdot C_{k,j}   & \text{ if } \rhs(A)=BC  \end{cases} 
\end{equation*}
In the first line one has to notice that $M[i,j]$ is one of the numbers $-1,0,1$.
The second line  is simply the rule for matrix multiplication ($A_{i,j} = \sum_{k=1}^d B_{i,k} C_{k,j}$) taking into
account that all matrices are unitriangular.

Now, $\val(\mathcal{G})$ is the identity matrix if and only if all matrix entries $\val_{\mathcal{G}}(S)[i,j]$ ($1 \leq i < j \leq d$) are zero.
But this is the case if and only if the sum of squares $\sum_{1 \leq i < j \leq d} \val_{\mathcal{G}}(S)[i,j]^2$ is zero.
Hence, we finally define
$$
\rhs(T) = \sum_{1 \leq i < j \leq d} S_{i,j}^2 .
$$
Concerning the multiplication depth, note that the multiplication depth of the gate $A_{i,j}$ is bounded by $j-i$:
The only multiplications in $\rhs(A_{i,j})$ are of the form $B_{i,k} C_{k,j}$ (and these multiplications are not nested).
Hence, by induction, the multiplication depth of $A_{i,j}$ is bounded by $1+ \max \{ k-i, j-k \mid i < k < j \} = j-i$.
It follows that every gate $S_{i,j}$ has multiplication depth at most $d-1$, which implies that the output gate $T$ has multiplication
depth at most $d$.

Similarly, it can be shown by induction that  $\mathsf{deg}(A_{i,j}) \leq j-i$. Hence, 
$\mathsf{deg}(A_{i,j}) \leq d-1$ for all $1 \leq i < j \leq d$, which implies that the formal degree
of the circuit is bounded by $2(d-1)$.

The structure-preserving partition $\biguplus_{i=1}^d V_i$ of the set of all multiplication gates of $\mathcal{C}$ can be defined as follows:
All gates corresponding to multiplications $B_{i,k} \cdot C_{k,j}$ in $\rhs(A_{i,j})$ are put into the set $V_{j-i}$.
Finally, all gates corresponding to multiplications  $S_{i,j}^2$ in $\rhs(T)$ are put into $V_d$. It is obvious 
that this partition is structure-preserving. 
\qed
\end{proof}
In a second step we apply Lemma~\ref{lemma-eliminate-sub} and construct from the above circuit $\mathcal{C}$ 
two variable-free positive circuits $\mathcal{C}_1$ and $\mathcal{C}_2$, both having multiplication depth
at most $d$ such that $\val(\mathcal{C}) = \val(\mathcal{C}_1) - \val(\mathcal{C}_2)$. Hence, our input SLP
$\mathcal{G}$  evaluates to the indentity matrix if and only if $\val(\mathcal{C}_1) = \val(\mathcal{C}_2)$.
Moreover, using the construction from Lemma~\ref{lemma-eliminate-sub}  it is straightforward to 
compute in logspace a structure-preserving partition $\biguplus_{i=1}^d V_{k,i}$ of the the set of all multiplication gates of $\mathcal{C}_k$ 
($k \in \{1,2\}$).

The following lemma concludes the proof that $\CWP(\UT_d(\mathbb{Z}))$ belongs to $\CL$.

\begin{lemma}
Let $d$ be constant. From a given variable-free positive circuit $\mathcal{C}$ of multiplication depth $d$ together
with a structure-preserving partition $\biguplus_{i=1}^d V_i$ of the set of all multiplication gates of $\mathcal{C}$,
we can compute in logarithmic space a variable-free addition circuit $\mathcal{D}$ such that $\val(\mathcal{C}) = \val(\mathcal{D})$.
\end{lemma}

\begin{proof}
Let $\mathcal{C} = (V,S,\rhs)$ together
with the partition $V = \biguplus_{i=1}^d V_i$ as in the lemma. W.l.o.g. we can assume that there is a unique input gate
whose right-hand side is $0$ (resp., $1$) and we denote this gate simply with $0$ (resp., $1$).

Since $d$ is a constant, it suffices to construct in logarithmic space a 
variable-free positive circuit $\mathcal{C}' = (V',S,\rhs')$ of multiplication depth $d-1$ together
with a structure-preserving partition $V' = \biguplus_{i=1}^{d-1} V'_i$ of the set of all multiplication gates of $\mathcal{C}'$
such that $\val(\mathcal{C}) = \val(\mathcal{C}')$ (the composition of a constant number of logspace computations is again a logspace
computation).

To achieve the above goal, we eliminate in $\mathcal{C}$ all multiplication gates from $V_1$. Note that below these gates there
are not other multiplication gates. Then, we define the set $V'_i$ as $V_{i+1}$ for $1 \leq i \leq d-1$. 

Let $V_1  = \{ A_1, \ldots, A_m\}$ and assume that $\rhs_{\mathcal{C}}(A_i) = B_i \cdot C_i$.  
The set of gates of $\mathcal{C}'$ is
$$
V' = V \cup \{ A^{(i)} \mid A \in V, 1 \leq i \leq m\},
$$
i.e., we add $m$ copies of each gate to the circuit.
We define the right-hand side mapping as follows:
\begin{gather}
\rhs'(A) = \rhs(A) \text{ if } A \in V \setminus V_1 \\
\rhs'(A_i) =  B_i^{(i)} \text{ for } 1 \leq i \leq m \label{edge-to-A_i}\\
\rhs'(A^{(i)}) = B^{(i)} + C^{(i)} \text{ if } A \in V \text{ and } \rhs(A) = B + C \\
\rhs'(A^{(i)}) = 0 \text{ if } A \in V \text{ and } \rhs(A) = B \cdot C \label{copies-mult}\\
\rhs'(0^{(i)}) = 0 \\
\rhs'(1^{(i)}) = C_i \label{edge-to-1^i}
\end{gather}
Note that $A_i$ has only one incoming edge after this construction. To stick to our definition of arithmetic circuits,
we can make $A_i$ an addition gate, which gets another incoming edge from $0$, and similarly
for $1^{(i)}$ (the $i$-th copy of the unique $1$-gate).

The idea of the above construction is the following: Basically, we add $m$ many copies of the circuit $\mathcal{C}$. In these copies,
we do not need the multiplication gates\footnote{Actually, we only need in the $i$-th copy those nodes that belong to a path from
the unique $1$-gate to $B_i$. But we cannot compute the set of these nodes in logspace unless $\mathsf{L}=\mathsf{NL}$. Hence,
we put all nodes into the copy.} and since we do not want to introduce new multiplication gates, we set the right-hand side 
of a copy of a multiplication gate to $0$, see \eqref{copies-mult}.\footnote{This is an arbitrary choice; instead of $0$ we could have also taken $1$.}
Also notice that strictly below $A_i$ we only find addition gates and constants in the circuit $\mathcal{C}$. 
In particular, the value $\val_{\mathcal{C}}(B_i)$ is equal to the number of paths from the unique $1$-gate $1$ to $B_i$
and similarly for $C_i$. We want to assign to gate $A_i$ the product of these path numbers. For this,
we redirect the edges $(B_i, A_i)$ and $(C_i,A_i)$ of the multiplication gate for every $1 \leq i \leq m$ as follows:
The edge $(C_i,A_i)$ is replaced by the edge $(C_i,1^{(i)})$, see \eqref{edge-to-1^i}. Moreover, the edge 
$(B_i,A_i)$ is replaced by the edge $(B_i^{(i)}, A_i)$ (which is the unique incoming edge to $A_i$), see \eqref{edge-to-A_i}.
So, basically, we serially connect the circuit part between $1$ and $C_i$ with the circuit part between $1$ and $B_i$.
Thereby we multiply the number of paths.
The above construction can be clearly done in logspace.
 \qed
\end{proof}
So far, we have restricted to \emph{torsion-free} f.g. nilpotent groups.
For general f.g. nilpotent groups, we use the fact that every f.g. nilpotent group contains 
a torsion-free normal f.g. nilpotent subgroup of finite index (Theorem~\ref{thm-nilpotent-subgroup})
in order to show that the compressed word problem for every f.g. nilpotent group belongs to 
the complexity class $\DET$: To do this we need the following result:  
  
\begin{theorem} \label{thm-NC1-red}
Let $G$ be a finitely generated group. For every normal subgroup $H$ of $G$ with a finite index, $\CWP(G)$ 
is $\AC^0$-reducible to $\CWP(H)$ and $\CWP(G / H)$.
\end{theorem}

\begin{proof}
To show the lemma, we adopt the proof of \cite[Theorem 4.4]{Loh14}, 
where the statement is shown for polynomial time many-one reducibility instead of $\AC^0$-reducibility.
Let $G$ be a finitely generated group with the finite generating set $\Sigma$ and let $H$ be a normal subgroup of $G$ of finite index (which must be f.g. as well)
with the finite generating set $\Gamma$.  As the generating set for the quotient $G/H$ we can take the set $\Sigma$ as well.
Let $\left\{ Hg_1,\ldots, Hg_n\right\}$ be the set of cosets of $H$ in $G$, where $g_1 = 1$. Moreover,
let  $\phi: G \to G/H$ be the canonical homomorphism and let $h: (\Sigma \cup \Sigma^{-1})^{*} \to G$ be the morphism that maps every word from 
$(\Sigma \cup \Sigma^{-1})^*$ to the 
group element in $G$ to which it evaluates. Now let $\mathcal{G} =(V, S, \rhs_{\mathcal{G}})$ be an SLP over the alphabet
$\Sigma \cup \Sigma^{-1}$. We have to construct an $\AC^0$-circuit with oracle gates for $\CWP(H)$ and $\CWP(G / H)$ that checks whether
$\val(\mathcal{G})=1$ in $G$. 

Consider the set of triples
\[W = \left\{ [g_i, A, g_j^{-1}] \mid A \in V, 1 \leq i,j \leq n, g_ih(\val_{ \mathcal{G}}(A))g_j^{-1} \in H \right\}. \] 
In a first step, we construct the set of all these triples using $n^2|V|$ parallel $\CWP(G/H)$-oracle gates.
More precisely, we construct for all $A \in V, 1 \leq i,j \leq n$ an SLP $\mathcal{G}_{A,i,j}$ that 
evaluates to the group element $\phi(g_ih(\val_{ \mathcal{G}}(A))g_j^{-1}) \in G/H$.
For this, we take the SLP $\mathcal{G}$ and add a new start variable $S_{A,i,j}$ with the right-hand side
$w_i A w_j^{-1}$, where $w_i \in (\Sigma \cup \Sigma^{-1})^{*}$ is a word that represents the group element $g_i$.
We do not need to compute these words $w_i$; they can be ``hard-wired'' into the circuit. 
The SLP $\mathcal{G}_{A,i,j}$ can be clearly constructed in $\AC^0$, and we have $\val(\mathcal{G}_{A,i,j}) = 1$ in $G/H$ 
if and only if $g_ih(\val_{ \mathcal{G}}(A))g_j^{-1} \in H$.

Note that $\val(\mathcal{G}_{S,1,1}) = 1$ in $G/H$ if and only if $\val(\mathcal{G})$ represents an element of the subgroup $H$.
Thus, if it turns out that $\val(\mathcal{G}_{S,1,1}) \neq 1$ in $G/H$, then the whole circuit will output zero.
Otherwise (i.e., in case $h(\val(\mathcal{G})) \in H$), we construct an SLP $\mathcal{H}$ over the alphabet $\Gamma \cup \Gamma^{-1}$ 
(the monoid generating set for $H$)
that will represent the group element $h(\val(\mathcal{G}))$.

The variable set of $\mathcal{H}$ is $W$, the start variable is $[g_1, S, g_1^{-1}]$  and the right-hand sides are defined as follows:
If $\rhs_{\mathcal{G}}(A)=a \in \Sigma \cup \Sigma^{-1}$,
we set $\rhs_{\mathcal{H}}([g_i,A,g_j^{-1}])=w_{a,i,j}$, where $w_{a,i,j} \in (\Gamma \cup \Gamma^{-1})^{*}$ is a word that represents 
the group element $g_iag_j^{-1} = g_ih(\val_{ \mathcal{G}}(A))g_j^{-1} \in H$. Note again, that we do not have to compute these words $w_{a,i,j}$
(they are fixed).
If $\rhs_{\mathcal{G}}(A)=BC$ and $[g_i, A, g_j^{-1}] \in W$, then
we determine the unique $k$, so that 
$g_ih(\val_{\mathcal{G}}(B))g_k^{-1} \in H$. To do this we have to go through the set $W$ and look for the unique $k$ such that $[g_i,B,g_k^{-1}] \in H$.  
Now we define $\rhs_{\mathcal{H}}([g_i, A, g_j^{-1}]) = [g_i,B,g_k^{-1}][g_k, C, g_j^{-1}]$. 
Clearly, this construction can be carried out by an $\AC^0$-circuit. 
Finally, it is straightforward to show that $\val_{\mathcal{H}}([g_i, A, g_j^{-1}])$ represents the group
element $g_ih(\val_{ \mathcal{G}}(A))g_j^{-1} \in H$. Hence, we have 
$\val(\mathcal{G}) = 1$ in $G$, if and only if  $\val(\mathcal{H})=1$ in $H$. This finishes our reduction.
Note that the overall circuit consists of $n^2|V|$ parallel $\CWP(G/H)$-oracle gates followed by a single  $\CWP(H)$-oracle gate.
\qed
\end{proof}
We can now show:

\begin{theorem} \label{thm-fg-nilpotent}
For every f.g. nilpotent group,
the compressed word problem is in $\DET$.
\end{theorem}

\begin{proof}
Let $G$ be a f.g. nilpotent group. If $G$ is finite, then the result follows from Theorem~\ref{thm-CWP-finite-groups}
(every nilpotent group is solvable).  If $G$ is infinite, then by Theorem~\ref{thm-nilpotent-subgroup}, 
$G$ has a torsion-free normal subgroup $H$ of finite index. By Theorem~\ref{thm-nilpotent-torsion-free-subgroup},
$H$ and $G/H$ are nilpotent too; moreover $H$ is finitely generated. By Theorem~\ref{thm-nilpotent-main},
$\CWP(H)$ belongs to $\CL \subseteq \DET$. Moreover, by Theorem~\ref{thm-CWP-finite-groups} $\CWP(G/H)$ 
belongs to $\DET$ as well. Finally, Theorem~\ref{thm-NC1-red} implies that $\CWP(G)$ belongs to $\DET$.
\qed
\end{proof}
Actually, Theorem~\ref{thm-fg-nilpotent} can be slightly extended to groups that are (f.g. nilpotent)-by-(finite solvable)
(i.e., groups that have a normal subgroup, which is f.g. nilpotent, and where the quotient is finite solvable.
This follows from Theorem~\ref{thm-NC1-red} and the fact that the compressed word problem for a finite solvable
group belongs to $\DET$ (Theorem~\ref{thm-CWP-finite-groups}).

\section{The uniform compressed word problem for unitriangular groups} \label{sec-uniform}

For Lemma~\ref{lemma-UT-in-CL} it is crucial that the dimension $d$ is a constant. 
In this section, we consider a uniform variant of the compressed word problem for $\UT_d(\mathbb{Z})$. 
We denote this problem with $\CWP(\UT_*(\mathbb{Z}))$. The input consists of a unary encoded number $d$ 
and an SLP, whose terminal symbols are generators of $\UT_d(\mathbb{Z})$ or there inverses. Alternatively,
we can assume that the terminal symbols are arbitrary matrices from $\UT_d(\mathbb{Z})$ with binary encoded
entries (given such a matrix $M$, it is easy to construct an SLP over the generator matrices that produces $M$).
The question is whether the SLP evaluates to the identity matrix. We show that this problem is complete for 
the complexity class $\CLOGCFL$.

\begin{theorem}
The problem $\CWP(\UT_*(\mathbb{Z}))$ is complete for $\CLOGCFL$.
\end{theorem}

\begin{proof}
We start with the upper bound. Consider an SLP $\mathcal{G}$, whose terminal symbols are generators of $\UT_d(\mathbb{Z})$ or there inverses.
The dimension $d$ is clearly bounded by the input size. Consider the variable-free arithmetic circuit $\mathcal{C}$ constructed from $\mathcal{G}$
in Lemma~\ref{lemma-SLP->circuit} and let $\mathcal{C}_1$ and $\mathcal{C}_2$ be the two variable-free positive arithmetic circuits obtained from $\mathcal{C}$
using Lemma~\ref{lemma-eliminate-sub}. Then $\mathcal{G}$ evaluates to the identity matrix if and only if $\val(\mathcal{C}_1) = \val(\mathcal{C}_2)$. 
Moreover, the formal degrees $\deg(\mathcal{C}_1)$ and $\deg(\mathcal{C}_2)$ are bounded by $2(d-1)$, i.e., polynomially bounded in the input length.
Finally, we compose a logspace machine that computes from the input SLP $\mathcal{G}$ the circuit $\mathcal{C}_i$ with the NAuxPDA 
from Lemma~\ref{lemma-circuit-NAuxPDA} to get an NAuxPDA $\mathcal{P}_i$ such that the number of accepting computation paths of $\mathcal{P}_i$ on input 
$\mathcal{G}$ is exactly $\val(\mathcal{C}_i)$. Moreover, the running time of $\mathcal{P}_i$ on input 
$\mathcal{G}$ is bounded polynomially in $(2d-1) \cdot \depth(\mathcal{C}_i) \in O(d \cdot |\mathcal{G}|)$.

Let us now show that  $\CWP(\UT_*(\mathbb{Z}))$ is hard for $\CLOGCFL$.
Let $(\mathcal{C}_{1,n})_{n \geq 0}$ and $(\mathcal{C}_{2,n})_{n \geq 0}$
be two logspace-uniform families of positive arithmetic circuits of polynomially bounded size and formal degree.
Let $w = a_1 a_2 \cdots a_n \in \{0,1\}^n$ be an input for the circuits $\mathcal{C}_{1,n}$ and $\mathcal{C}_{2,n}$. 
Let $\mathcal{C}_i$ be the variable-free positive arithmetic circuit obtained from $\mathcal{C}_{i,n}$ by replacing every
$x_j$-labelled input gate by $a_j \in \{0,1\}$. 
By \cite[Lemma~3.2]{AllenderJMV98} we can assume that every gate of $\mathcal{C}_i$ is labelled by its formal degree.
By adding if necessary additional multiplication gates, where one input is set to $1$, we
can assume that $\mathcal{C}_1$ and $\mathcal{C}_2$ have 
the same formal degree $d \leq p(n)$ for a polynomial $p$. Analogously, we can assume that
if $A$ is an addition gate in $\mathcal{C}_1$ or $\mathcal{C}_2$ with right-hand side $B+C$,
then $\deg(B) = \deg(C) = \deg(A)$. All these preprocessing steps can be carried out in logarithmic space.

We will construct in logarithmic space 
an SLP $\mathcal{G}$ over the alphabet $\Gamma_{d+1} \cup \Gamma_{d+1}^{-1}$, where $\Gamma_{d+1}$ is our 
canonical generating set for the matrix group $\UT_{d+1}(\mathbb{Z})$, such
that $\mathcal{G}$ evaluates to the identity matrix if and only if $\mathcal{C}_1$ and  $\mathcal{C}_2$
evaluate to the same number.
Let $v_i$ be the output value of $\mathcal{C}_i$.
We first construct in logspace an SLP $\mathcal{G}_1$ that evaluates to the matrix $T_{1,d}^{v_1}$. In the same way we can construct
in logspace a second SLP $\mathcal{G}_2$ that evaluates to $T_{1,d}^{-v_2}$. Then, by concatenating the two SLPs
$\mathcal{G}_1$ and $\mathcal{G}_2$ we obtain the desired SLP.

The variables of $\mathcal{G}_1$ are $A_{i,j}^{b}$, where $A$ is a gate of $\mathcal{C}_1$,
$b \in \{-1,1\}$, and $1 \leq i < j \leq d$ such that  
$j-i$ is the formal degree of $A$.  The SLP $\mathcal{G}_1$ will be constructed in such a way that
$\val_{\mathcal{G}_1}(A_{i,j}^{b}) = T_{i,j}^{b \cdot v}$, where $v = \val_{\mathcal{C}_1}(A)$.
If $\rhs_{\mathcal{C}_1}(A) = 0$, then we set $\rhs_{\mathcal{G}_1}(A_{i,j}^b) = \Id$ and 
if $\rhs_{\mathcal{C}_1}(A) = 1$, then we set $\rhs_{\mathcal{G}_1}(A_{i,j}^b) = T_{i,j}^b$.
Correctness is obvious in these cases.
If $\rhs_{\mathcal{C}_1}(A) = B+C$, then we set $\rhs_{\mathcal{G}_1}(A_{i,j}^b) = B_{i,j}^b C_{i,j}^b$.
Correctness follows immediately by induction. Note that $\deg(B) = \deg(C) = \deg(A) = j-i$, which implies
that the gates $B_{i,j}^b$ and $C_{i,j}^b$ exist.
Finally, if $\rhs_{\mathcal{C}_1}(A) = B \cdot C$, then we set 
$\rhs_{\mathcal{G}_1}(A_{i,j}^1) = B_{i,k}^{-1} C_{k,j}^{-1} B_{i,k}^1 C_{k,j}^1$ and 
$\rhs_{\mathcal{G}_1}(A_{i,j}^{-1}) = C_{k,j}^{-1} B_{i,k}^{-1} C_{k,j}^{1} B_{i,k}^{1}$,
where $k$ is such that $\deg(B) = k-i$ and $\deg(B) = j-k$. Such a $k$ must exist since $j-i = \deg(A) = \deg(B)+\deg(C)$.
Correctness follows from Lemma~\ref{lemma-product} and induction.
\qed
\end{proof}

\section{The compressed word problem for polycyclic groups} \label{sec-polycyclic}

In this section we consider the compressed word problem for polycyclic groups. Since every polycyclic group is f.g. linear, the compressed
word problem for a polycyclic group can be reduced to polynomial identity testing. In this section, we show a lower bound:
There exists a strongly polycyclic group $G$ (which is also metabelian) such that polynomial identity testing for skew arithmetic circuits can be reduced to $\CWP(G)$. 

Let us start with a specific example of a polycyclic group. 
Consider the two matrices 
\begin{equation}  \label{g_a-and-h}
g_a =\left( \begin{array}{cc} a \ & 0 \\ 0 & 1 \end{array} \right) \text{ and }  h = \left( \begin{array}{cc} 1 \ & 1 \\ 0 & 1 \end{array} \right),
\end{equation}
where $a \in \mathbb{R}$, $a \geq 2$. Let $G_a = \langle g_a, h \rangle \leq \mathsf{GL}_2(\mathbb{R})$.
Let us remark that, for instance, the group $G_2$ 
is not polycyclic, see e.g. \cite[p.~56]{Wehr73}. On the other hand, we have:

\begin{proposition} \label{prop-a-polycyclic-group}
The group $G = G_{1+\sqrt{2}}$ is polycyclic and metabelian.\footnote{It is probably known to experts that $G$
is polycyclic. Since we could not find an explicit proof, we present the arguments for completeness.}
\end{proposition}

\begin{proof}
We show that the commutator subgroup of $G$ is isomorphic to $\mathbb{Z} \times \mathbb{Z}$, which implies  the theorem.
First we calculate the commutator subgroup of $G$. It is known that the commutator subgroup of a group generated by two elements $g_1,g_2$ is 
generated by all commutators $g_1^{s}g_2^{t}g_1^{-s}g_2^{-t}$ for  $s,t \in \mathbb{Z}$ \cite{Mil32}. Hence, 
$$[G,G] = \langle M_{s,t} \mid s,t \in \mathbb{Z} \rangle,$$ 
where for $s,t \in \mathbb{Z}$ we set
\begin{eqnarray*}
M_{s,t} & =  & \left( \begin{array}{cc} 1+ \sqrt{2} \ & 0 \\ 0 & 1 \end{array} \right)^s \left( \begin{array}{cc} 1 & 1 \\ 0  & 1 \end{array} \right)^t \left( \begin{array}{cc} 1+ \sqrt{2} \ & 0 \\ 0 \ & 1 \end{array} \right)^{-s} \left( \begin{array}{cc} 1 & 1 \\ 0 & 1 \end{array} \right)^{-t} \\
& = & \left( \begin{array}{cc} (1+ \sqrt{2})^s \ & 0 \\ 0 & 1 \end{array} \right) \left( \begin{array}{cc} 1 & t \\ 0 \ & 1 \end{array} \right) \left( \begin{array}{cc} (1+ \sqrt{2})^{-s} \ & 0 \\ 0 \ & 1 \end{array} \right) \left( \begin{array}{cc} 1 & -t \\ 0 & 1 \end{array} \right) \\
& =&  \left( \begin{array}{cc}  (1+ \sqrt{2})^{s} \ & t (1+ \sqrt{2})^{s} \\ 0 \ & 1 \end{array} \right) \left( \begin{array}{cc}  (1+ \sqrt{2})^{-s} \ & -t (1+ \sqrt{2})^{-s} \\ 0 \ & 1 \end{array} \right)  \\
& = & \left( \begin{array}{cc} 1 \ & -t + t(1+\sqrt{2})^s \\ 0 \ & 1 \end{array} \right)  \\
&  = &  \left( \begin{array}{cc} 1 \ &  t((1+\sqrt{2})^s-1) \\ 0 \ & 1 \end{array} \right) .
\end{eqnarray*} 
With the setting 
$$u= \left( \begin{array}{cc} 1 \ &  \sqrt{2} \\ 0 & 1 \end{array} \right) \quad \text{and} \quad v=\left( \begin{array}{cc} 1 \  &  2 \\ 0 & 1 \end{array} \right)
$$
we show that $\langle M_{s,t} \mid s,t \in \mathbb{Z} \rangle = \langle u,v \rangle$.
Moreover, it is easy to see that $u$ and $v$ generate a copy of $\mathbb{Z} \times \mathbb{Z}$.
 
We have $M_{1,1} = u$ and
$$
M_{2,1} M_{1,1}^{-2} =  \left( \begin{array}{cc} 1 \ & 2+ 2\sqrt{2} \\ 0 & 1 \end{array} \right) 
 \left( \begin{array}{cc} 1 \ & - 2\sqrt{2} \\ 0 & 1 \end{array} \right) = \left( \begin{array}{cc} 1 \ &  2 \\ 0 & 1 \end{array} \right) = v.
$$
This shows that $\langle u,v \rangle \subseteq \langle M_{s,t} \mid s,t \in \mathbb{Z} \rangle$.
For the other inclusion assume first that
$s \geq 0$. Then 
\begin{eqnarray*}
t\left( \left(1+\sqrt{2}\right)^s-1\right) &=& t\left( \left( \sum_{i=0}^s {s \choose i} \sqrt{2}^i \right) -1\right)\\
&= & t\left(\sum_{i=1}^s {s \choose i} \sqrt{2}^i\right) \\
&= & t\left(\sum_{i=1}^{\lfloor {s \over 2} \rfloor} {s \choose 2i} \sqrt{2}^{2i} + \sum_{i=1}^{\lceil {s \over 2} \rceil} {s  \choose 2i-1} \sqrt{2}^{2i-1}\right)\\
& =& 2\left(\sum_{i=1}^{\lfloor {s \over 2} \rfloor}t {s \choose 2i} 2^{i-1}\right)+ \sqrt{2}\left( \sum_{i=1}^{\lceil {s \over 2} \rceil} t {s  \choose 2i-1} 2^{i-1}\right).
 \end{eqnarray*}
So with 
$$
c_1=\sum_{i=1}^{\lfloor {s \over 2} \rfloor}t {s \choose 2i} 2^{i-1} \in \mathbb{Z} \quad\text{and}\quad
c_2=\sum_{i=1}^{\lceil {s \over 2} \rceil} t {s  \choose 2i-1} 2^{i-1} \in \mathbb{Z}
$$ 
we get
$$
M_{s,t} = \left( \begin{array}{cc} 1 \ & t((1+\sqrt{2})^s-1) \\ 0 \ & 1 \end{array} \right)=\left( \begin{array}{cc} 1 \ & 2c_1 + \sqrt{2} c_2 \\ 0 \ & 1 \end{array} \right) = v^{c_1}u^{c_2} .
$$
For $s < 0$ we get with $a = -s \text{ mod } 2$:
\begin{eqnarray*}
t\left( \left(1+\sqrt{2}\right)^s-1\right) &=& t\left( \left(\sqrt{2} - 1\right)^{-s}-1\right)\\
&=& t\left( \left(\sum_{i=0}^{-s} {-s \choose i} (\sqrt{2})^i (-1)^{-s-i} \right)  -1\right)\\
&=& t\left(-2a + \sum_{i=1}^{-s} {-s \choose i} (\sqrt{2})^i (-1)^{-s-i} \right)\\
&=& t\left( -2a+ \sum_{i=1}^{\lfloor {-s \over 2} \rfloor} {-s \choose 2i} (\sqrt{2})^{2i} (-1)^{-s-2i}\right) + \\
& &  t\sum_{i=1}^{\lceil {-s \over 2} \rceil} {-s  \choose 2i-1} (\sqrt{2})^{2i-1} (-1)^{-s-(2i-1)}\\
&=& 2\left( -at+\sum_{i=1}^{\lfloor {-s \over 2} \rfloor}t {-s \choose 2i} 2^{i-1} (-1)^{-s-2i}\right) + \\
& &  \sqrt{2}\left( \sum_{i=1}^{\lceil {-s \over 2} \rceil} t {-s  \choose 2i-1} 2^{i-1} (-1)^{-s-(2i-1)}\right).
\end{eqnarray*}
 So with 
 \[c_1=-at+\sum_{i=1}^{\lfloor {-s \over 2} \rfloor}t {-s \choose 2i} 2^{i-1} (-1)^{-s-2i} \in \mathbb{Z}\] and
 \[c_2=\sum_{i=1}^{\lceil {-s \over 2} \rceil} t {-s  \choose 2i-1} 2^{i-1} (-1)^{-s-(2i-1)} \in \mathbb{Z}\] we get 
 \[M_{s,t} = \left( \begin{array}{cc} 1 \ & t((1+\sqrt{2})^s-1) \\ 0 \ & 1 \end{array} \right)=v^{c_1}u^{c_2} .\]
This shows that $\langle M_{s,t} \mid s,t \in \mathbb{Z} \rangle \subseteq \langle u,v \rangle$.
\qed
\end{proof}
The main result of this section is:

\begin{theorem} \label{thm-polycyclic}
Let $a \geq 2$.
Polyomial  identity testing for skew arithmetic circuits is logspace-reducible to the compressed word problem for the group $G_a$.
\end{theorem}
In particular, there exist polycyclic groups for which the compressed word problem is at least as hard as 
polynomial  identity testing for skew circuits.  Recall that it is not known, whether there exists a 
polynomial time algorithm for polynomial identity testing restricted to skew arithmetic circuits.

For the proof of Theorem~\ref{thm-polycyclic}, we will make use of the following result from \cite{AllenderBKM09} (see the proof of 
Proposition~2.2 in  \cite{AllenderBKM09}, where the result is shown for $a=2$, but the proof immediately generalizes to any $a \geq 2$):

\begin{lemma} \label{lemma-allender}
Let $\mathcal{C}$ be an arithmetic circuit of size $n$ with variables $x_1, \ldots, x_m$ and let
$p(x_1, \ldots, x_m) = \val(\mathcal{C})$. Let $a \geq 2$ be a real number.
Then $p(x_1,\ldots,x_n)$ is the zero-polynomial if
and only if $p(\alpha_1,\ldots,\alpha_n)=0$, where $\alpha_i=a^{2^{i \cdot n^2}}$ for $1 \leq i \leq m$.
\end{lemma}
{\it Proof of Thereom~\ref{thm-polycyclic}.}
Let us fix a skew arithmetic circuit $\mathcal{C}$ of size $n$ with $m$ variables $x_1, \ldots, x_m$. 
We will  define an SLP $\mathcal{G}$ over the alphabet $\{ g_a, g^{-1}_a, h, h^{-1} \}$
such that $\val(\mathcal{G})=\Id$ in $G_a$ if and only if $\val(\mathcal{C}) = 0$.
First of all, using iterated squaring, we can construct an SLP $\mathcal{H}$ with variables $A_1, A_1^{-1} \ldots, A_m, A_m^{-1}$ (and 
some other auxiliary variables)
such that
\begin{eqnarray*}
\val_{\mathcal{H}}(A_i) & = & g_a^{2^{i \cdot n^2}} = \left( \begin{array}{cc} a^{2^{i \cdot n^2}} \ & 0 \\ 0 & 1 \end{array} \right) = 
\left( \begin{array}{cc} \alpha_i \ & 0 \\ 0 & 1 \end{array} \right) \text{ and } \\
\val_{\mathcal{H}}(A_i^{-1}) & = & g_a^{-2^{i \cdot n^2}} = \left( \begin{array}{cc} a^{- 2^{i \cdot n^2}} \ & 0 \\ 0 & 1 \end{array} \right) = 
\left( \begin{array}{cc} \alpha_i^{-1} \ & 0 \\ 0 & 1 \end{array} \right) .
\end{eqnarray*}
We now construct the SLP $\mathcal{G}$ as follows: The set of variables of $\mathcal{G}$ consists of the gates
of  $\mathcal{C}$ and the variables of $\mathcal{H}$. We copy the right-hand sides from $\mathcal{H}$ and define
the right-hand side for a gate $A$ of $\mathcal{C}$ as follows:
\begin{equation*} \label{rhs-polycyclic}
\rhs_{\mathcal{G}}(A) = 
\begin{cases} 
  \Id & \text{ if } \rhs_{\mathcal{C}}(A) = 0 \\ 
  h  & \text{ if } \rhs_{\mathcal{C}}(A) = 1 \\
  h^{-1} & \text{ if } \rhs_{\mathcal{C}}(A) = -1  \\
  B C  & \text{ if } \rhs_{\mathcal{C}}(A) = B+C \\
  A_i B A_i^{-1}   & \text{ if } \rhs_{\mathcal{C}}(A) = x_i \cdot B
\end{cases}
\end{equation*}
We claim that for every gate $A$ of $\mathcal{C}$ we have the following, where we denote for better readability the polynomial $\val_{\mathcal{C}}(A)$
to which gate $A$ evaluates with $p_A$:
$$
\val_{\mathcal{G}}(A) =  \left( \begin{array}{cc} 1 \ & p_A(\alpha_1, \ldots, \alpha_n) \\ 0 \ & 1 \end{array} \right)
$$
The case that $\rhs_{\mathcal{C}}(A)$ is a constant is obvious. If $\rhs_{\mathcal{C}}(A) = B+C$ then we obtain by induction
\begin{eqnarray*}
\val_{\mathcal{G}}(A) &=& \val_{\mathcal{G}}(B) \val_{\mathcal{G}}(C) \\
&=& \left( \begin{array}{cc} 1 \ & p_B(\alpha_1, \ldots, \alpha_n) \\ 0 \ & 1 \end{array} \right) \left( \begin{array}{cc} 1 \ & p_C(\alpha_1, \ldots, \alpha_n) \\ 0 \ & 1 \end{array} \right) \\
&=& \left( \begin{array}{cc} 1 \ & p_B(\alpha_1, \ldots, \alpha_n) + p_C(\alpha_1, \ldots, \alpha_n) \\ 0 \ & 1 \end{array} \right) \\
&=& \left( \begin{array}{cc} 1 \ & p_A(\alpha_1, \ldots, \alpha_n)  \\ 0 \ & 1 \end{array} \right) .
\end{eqnarray*}
Finally, if $\rhs_{\mathcal{C}}(A) = x_i \cdot B$ then we obtain by induction
\begin{eqnarray*}
\val_{\mathcal{G}}(A) &=&
\left( \begin{array}{cc} \alpha_i \ & 0 \\ 0 & 1 \end{array} \right) \val_{\mathcal{G}}(B) \left( \begin{array}{cc} \alpha_i^{-1} \ & 0 \\ 0 & 1 \end{array} \right) \\
&=& \left( \begin{array}{cc} \alpha_i \ & 0 \\ 0 & 1 \end{array} \right) \left( \begin{array}{cc} 1 \ & p_B(\alpha_1, \ldots, \alpha_n) \\ 0 \ & 1 \end{array} \right) \left( \begin{array}{cc} \alpha_i^{-1} \ & 0 \\ 0 & 1 \end{array} \right) \\
& = & \left( \begin{array}{cc} \alpha_i \ & \alpha_i \cdot p_B(\alpha_1, \ldots, \alpha_n) \\ 0 \ & 1 \end{array} \right) \left( \begin{array}{cc} \alpha_i^{-1} \ & 0 \\ 0 & 1 \end{array} \right) \\
& = & \left( \begin{array}{cc} 1 \ & \alpha_i \cdot p_B(\alpha_1, \ldots, \alpha_n) \\ 0 \ & 1 \end{array} \right) \\
& = & \left( \begin{array}{cc} 1 \ & p_A(\alpha_1, \ldots, \alpha_n) \\ 0 \ & 1 \end{array} \right) .
\end{eqnarray*}
We finally take the output gate $S$ of the skew circuit $\mathcal{C}$ as the start variable of $\mathcal{G}$.
Then, $\val(\mathcal{G})$ yields the identity matrix in the group $G_a$ if and only if $p_S(\alpha_1, \ldots, \alpha_n) = 0$.
By Lemma~\ref{lemma-allender} this is equivalent to $\val(\mathcal{C}) = p_S(x_1,\ldots,x_n) = 0$.
\qed

\medskip
\noindent
Actually, we can carry out the above reduction for a class of arithmetic circuits that is slightly larger than 
the class of skew arithmetic circuits. Let us define a {\em powerful skew circuit} as an arithmetic
circuit, where for every multiplication gate $A$, $\rhs(A)$ is of the form $\alpha \cdot \prod_{i=1}^m x_i^{e_i} \cdot B$
for a gate $B$, binary coded integers $\alpha, e_1, \ldots, e_m$ ($e_i \geq 0$), and variables $x_1, \ldots, x_m$.
Such a circuit can be converted into an ordinary arithmetic circuit, which, however is no longer skew.
To extend the reduction from the proof of Thereom~\ref{thm-polycyclic} to powerful skew circuits,
first note that in a right-hand side $\alpha \cdot \prod_{i=1}^m x_i^{e_i} \cdot B$ we can assume that
$\alpha=1$, since we can obtain $\alpha \cdot \prod_{i=1}^m x_i^{e_i} \cdot B$ from $\prod_{i=1}^m x_i^{e_i} \cdot B$
using additional addition gates. For a gate $A$ with $\rhs_{\mathcal{C}}(A) = \prod_{i=1}^m x_i^{e_i} \cdot B$ we set
 $\rhs_{\mathcal{G}}(A)  = \prod_{i=1}^m A_i^{e_i} B \prod_{i=1}^m A_i^{-e_i}$.
The powers $A_i^{e_i}$ and $A_i^{-e_i}$ can be defined using additional multiplication gates.
In our recent paper \cite{KoeLo15power}, we introduced powerful skew circuits, and proved that for this class,
polynomial identity testing can be solved in $\mathsf{coRNC}$. We applied this result to the compressed
word problem for wreath products.

Let us look again at the group $G=G_{1+\sqrt{2}}$ from Proposition~\ref{prop-a-polycyclic-group}. Its commutator subgroup is isomorphic 
to $\mathbb{Z} \times \mathbb{Z}$. Moreover, the quotient $G / [G,G]$ is isomorphic to $\mathbb{Z} \times \mathbb{Z}_2$: 
The $G$-generator $h$ from \eqref{g_a-and-h} satisfies $h^2 \in [G,G]$, whereas the generator $g_{1+\sqrt{2}}$ has infinite order in the quotient.
Hence, $G$ has a subnormal series of the form $G \rhd H \rhd \mathbb{Z} \times \mathbb{Z} \rhd \mathbb{Z} \rhd 1$, where
$H$ has index $2$ in $G$ and $H / (\mathbb{Z} \times \mathbb{Z}) \cong \mathbb{Z}$. 
The group $H$ is strongly polycyclic and has Hirsch length $3$.
By Theorem~\ref{thm-NC1-red} we obtain:

\begin{corollary}
There is a strongly polycyclic group $H$ of Hirsch length 3 such that
polynomial identity testing for skew circuits is polynomial time reducible to $\CWP(H)$.
\end{corollary}


\def\cprime{$'$} \def\cprime{$'$}

\appendix

\section{The complexity of the classical word problem for  finitely generated linear groups}

In this section we consider the ordinary
(uncompressed) word problem for linear groups. 
The most important result in this context was shown by Lipton and Zalcstein \cite{LiZa77}:

\begin{theorem}
For every f.g. linear group the word problem can be solved in deterministic logarithmic space.
\end{theorem}
By Tits alternative \cite{Tits72}, every linear group is either virtually solvable (i.e., has a solvable subgroup of finite index, which can be assumed to be normal) or contains a free group of rank $2$.
Since by \cite[Theorem~6.3]{Rob93}, the word problem for a free group of rank $2$ is hard for $\DLOGTIME$-uniform $\NC^1$,
one gets:

\begin{theorem}
For every f.g linear group that is not virtually solvable, the word problem is hard for $\DLOGTIME$-uniform $\NC^1$.
\end{theorem}
This leads to the question for the complexity of the word problem for a virtually solvable linear group.
For the special case of a polycyclic group, Robinson \cite[Theorem~8.5]{Rob93} proved that the word problem belongs to $\TC^0$, but 
his circuits are not uniform. Waack proved in \cite{Waa91} that the word problem for a virtually solvable linear group
belongs to logspace-uniform $\NC^1$. Using the famous division breakthrough by Hesse et al. \cite{HeAlBa02}, we can improve 
Waack's result in the following way:

\begin{theorem} \label{thm-wp-linear-solvable}
For every f.g. virtually solvable linear group $G$ the word problem belongs to $\DLOGTIME$-uniform $\NC^1$.
If $G$ is moreover infinite solvable, then the word problem is complete for $\DLOGTIME$-uniform $\TC^0$.
\end{theorem}
For the proof, we first have to consider the complexity of iterated multiplication and division with remainder for 
polynomials in several variables. Recall that $\mathbb{Z}[x_1,\ldots,x_k]$ denotes the ring of polynomials in the variables
$x_1,\ldots,x_k$ with coefficients from $\mathbb{Z}$. 
For a polynomial $p \in \mathbb{Z}[x_1,\ldots,x_k]$ and a variable $x_i$ we denote with 
$\deg_{x_i}(p)$ the maximal value $d$ such that $x_i^d$ appears in a monomial of $p$.
We specify polynomials from $\mathbb{Z}[x_1,\ldots,x_k]$ by writing down for every non-zero term
$a x_1^{n_1} \cdots x_k^{n_k}$ the tuple of integers $(a, n_1, \ldots, n_k)$, where $a$ is represented in binary notation
and the exponents are represented in unary notation. 
Iterated multiplication  of polynomials in the ring $\mathbb{Z}[x_1,\ldots,x_k]$ is the task of computing 
from a given list of polynomials $p_1, p_2, \ldots, p_n \in \mathbb{Z}[x_1,\ldots,x_k]$ the product polynomial
$p_1 p_2 \cdots p_n$. Division with remainder in the ring $\mathbb{Z}[x]$ (later, we will generalize this to several variables)
is the task of computing for given polynomials $s,t \in \mathbb{Z}[x]$ such that $t \neq 0$ and the leading coefficient of $t$ is $1$
the unique polynomials $s \text{ mod } t$ and 
$s \text{ div } t$ such that $s = (s \text{ div } t) \cdot t + s \text{ mod } t$ and $\deg(s \text{ mod } t) < \deg(t)$, where
$\deg(p)$ denotes the degree of the polynomial $t$.

The following result was shown in \cite{Eberly89,HeAlBa02}:\footnote{Explicitly, the result is stated in \cite[Corollary~6.5]{HeAlBa02},
where the authors note that Eberly's reduction \cite{Eberly89} from iterated polynomial multiplication to iterated integer multiplication is actually
an $\AC^0$-reduction, which yields a $\DLOGTIME$-uniform $\TC^0$ bound with the main result from \cite{HeAlBa02}.}

\begin{theorem}[c.f.~\cite{Eberly89,HeAlBa02}] \label{lemma-eberly}
Iterated multiplication and division with remainder of polynomials in the ring $\mathbb{Z}[x]$ (respectively, $\mathbb{F}_p[x]$)
belong to $\DLOGTIME$-uniform $\TC^0$.
\end{theorem}
We need generalizations of Lemma~\ref{lemma-eberly} to multivariate polynomials. In the following proofs we always use
the fact that iterated addition, iterated multiplication and division with remainder of binary coded integers can be done in $\DLOGTIME$-uniform $\TC^0$
\cite{HeAlBa02}.

\begin{lemma} \label{lemma-poly-mult-TC0}
Iterated multiplication  of polynomials in the ring $\mathbb{Z}[x_1,\ldots,x_k]$ (respectively, $\mathbb{F}_p[x_1,\ldots,x_k]$) belongs to 
$\DLOGTIME$-uniform $\TC^0$.
\end{lemma}

\begin{proof}
We only prove the result for $\mathbb{Z}[x_1,\ldots,x_k]$; exactly the same proof also works for $\mathbb{F}_p[x_1,\ldots,x_k]$.
 
For $d \geq 1$ let $\mathbb{Z}[x_1,\ldots,x_k]_d \subseteq \mathbb{Z}[x_1,\ldots,x_k]$ be the set of 
all polynomials $p \in \mathbb{Z}[x_1,\ldots,x_k]$ such that 
$\deg_{x_i}(p) \leq d$ for all $1 \leq i \leq k$. For $d \geq 2$
we define the mapping $\mathcal{U}_d: \mathbb{Z}[x_1,\dots,x_k] \to \mathbb{Z}[z]$
by
$$
\mathcal{U}_d(p(x_1, x_2,\dots,x_k)) = p(z^{d^{0}},z^{d^{1}},\ldots,z^{d^k}).
$$
The mapping $\mathcal{U}_{d}$ is also used in \cite{AgrawalB03} to reduce polynomial identity testing to 
univariate polynomial identity testing.
The mapping $\mathcal{U}_{d+1}$ restricted to $\mathbb{Z}[x_1,\ldots,x_k]_d$ is injective, since
for a polynomial $p \in \mathbb{Z}[x_1,\ldots,x_k]_d$ we obtain the polynomial $\mathcal{U}_{d+1}(p)$ by replacing for every
monomial $a \cdot x_1^{n_1} \cdots x_k^{n_k}$ by the monomial $a \cdot z^N$, where $N$ the number with base-$(d+1)$
expansion $(n_1 \cdots n_k)$ (with the most significant digit on the right).
Moreover, for  all polynomials $p,q \in \mathbb{Z}[x_1,\ldots,x_k]$ and all $d \geq 2$ we have 
\begin{equation} \label{U-is-homomorph}
\mathcal{U}_d(p + q) = \mathcal{U}_d(p) + \mathcal{U}_d(q)
\text{ and }\mathcal{U}_d(p q)= \mathcal{U}_d(p)\mathcal{U}_d(q).
\end{equation}
We can calculate $\mathcal{U}_d(p)$ for a given polynomial  $p \in \mathbb{Z}[x_1,\ldots,x_k]$ and a given number $d \geq 2$ 
in $\DLOGTIME$-uniform $\TC^0$: For a monomial $a x_1^{n_1} \cdots x_k^{n_k}$ (which is represented
by the tuple $(a,n_1, \ldots, n_k)$) we have to compute the pair $(a, \sum_{i=0}^{k-1} n_{i+1} d^i)$,
which is possible in $\DLOGTIME$-uniform $\TC^0$.
Similarly, we can compute 
$\mathcal{U}^{-1}_{d+1}(p)$ for a polynomial $p \in \mathcal{U}_{d+1}(\mathbb{Z}[x_1,\ldots,x_k]_d)$.
in $\DLOGTIME$-uniform $\TC^0$: From a given monomial $a z^N$ (represented by the pair $(a,N)$) we have
to compute the tuple $(a, n_1, \ldots, n_k)$, where $n_i = (N \text{ div } (d+1)^{i-1}) \text{ mod } (d+1)$,
which can be done in $\DLOGTIME$-uniform $\TC^0$.

We now multiply given polynomials  $p_1,\dots,p_n \in \mathbb{Z}[x_1,\dots,x_k]$ in the following way, where
all steps can be carried out in $\DLOGTIME$-uniform $\TC^0$ by the above remarks.
\begin{enumerate}
\item 
Compute the number $d  = \max\{ \sum_{i=1}^n \deg_{x_j}(p_i) \mid  1 \leq j \leq k\}$.
This number bounds the degree of the product polynomial $p_1 p_2 \cdots p_n$ in any of the variables
$x_1, \ldots, x_n$, i.e., $p_1 p_2 \cdots p_n \in \mathbb{Z}[x_1,\ldots,x_k]_d$.
\item 
Compute in parallel $s_i(z) = \mathcal{U}_{d+1}(p_i)$ for $1 \leq i \leq n$.
\item Using Theorem~\ref{lemma-eberly}, compute the product
$S(z) = s_1(z) s_2(z) \cdots s_n(z)$, which is $\mathcal{U}_{d+1}(p_1 p_2\cdots p_n)$ by \eqref{U-is-homomorph}.
\item 
Finally, compute $\mathcal{U}^{-1}_{d+1}(S)$, which is $p_1 p_2 \cdots p_n$. 
\qed
\end{enumerate}
\end{proof}
For polynomial division in several variables,  we need a distinguished variable.
Therefore, we consider the polynomial ring $\mathbb{Z}[x_1,\ldots,x_k,y]$.
We view polynomials from this ring
as polynomials in the variable $y$, where coefficients are polynomials from $\mathbb{Z}[x_1,\ldots,x_k]$.
We will only divide by a polynomial $t$ for which the leading monomial $p(x_1, \ldots, x_n) y^m$ of $t$
satisfies $p(x_1, \ldots, x_n) = 1$. This ensures that the coefficients of the quotient and remainder polynomial are again
in $\mathbb{Z}[x_1,\ldots,x_k]$ (and not in the quotient field $\mathbb{Q}(x_1,\ldots,x_n)$). 

\begin{lemma} \label{lemma-poly-div-TC0}
Division with remainder of polynomials in the ring $\mathbb{Z}[x_1,\ldots,x_k,y]$ (respectively, $\mathbb{F}_p[x_1,\ldots,x_k,y]$) belongs to
$\DLOGTIME$-uniform $\TC^0$.
\end{lemma}

\begin{proof}
Again, we only prove the result for $\mathbb{Z}[x_1,\ldots,x_k,y]$; exactly the same proof works for $\mathbb{F}_p[x_1,\ldots,x_k,y]$ as well.
As in the proof of Lemma~\ref{lemma-poly-mult-TC0} consider the set $\mathbb{Z}[x_1,\ldots,x_k,y]_d \subseteq \mathbb{Z}[x_1,\ldots,x_k,y]$
of all polynomials in $\mathbb{Z}[x_1,\ldots,x_k,y]$ such that for every monomial $a \cdot x_1^{n_1} \cdots x_k^{n_k} y^n$ we have
$n_1, \ldots, n_k,n < d$, and the mapping $\mathcal{U}_d: \mathbb{Z}[x_1,\dots,x_k,y] \to \mathbb{Z}[z]$ with
$$
\mathcal{U}_d(p(x_1,x_2, \dots,x_k,y)) = p(z^{d^{0}},z^{d^{1}},\dots,z^{d^{k-1}},z^{d^k}).
$$
Note that for  polynomials $p,q \in \mathbb{Z}[x_1,\ldots,x_k,y]_d$ with $\deg_y(p) < \deg_y(q)$ we have $\deg(\mathcal{U}_{d+1}(p)) < \deg(\mathcal{U}_{d+1}(q))$,
since the exponent of $y$ becomes the most significant digit in the base-$(d+1)$ representation.
Then, for all polynomials $s,t \in \mathbb{Z}[x_1,\ldots,x_k,y]_d$ (where the leading coefficient of $t$ is $1$)
we have 
$$
\mathcal{U}_{d^2+1}(s \text{ mod } t)= \mathcal{U}_{d^2+1}(s) \text{ mod } \mathcal{U}_{d^2+1}(t).
$$ 
To see this, assume that
$s = q t + r$ with $\deg_y(r) < \deg_y(t)$, so that $r = s \text{ mod } t$.
We have $q,r \in \mathbb{Z}[x_1,\ldots,x_k,y]_{d^2}$, which can be checked by tracing the polynomial
division algorithm.
By \eqref{U-is-homomorph} we have 
$$\mathcal{U}_{d^2+1}(s) = \mathcal{U}_{d^2+1}(q) \mathcal{U}_{d^2+1}(t) + \mathcal{U}_{d^2+1}(r).
$$
Moreover,  $\deg(\mathcal{U}_{d^2+1}(r)) < \deg(\mathcal{U}_{d^2+1}(t))$. Hence 
$$
\mathcal{U}_{d^2+1}(r) = \mathcal{U}_{d^2+1}(s) \text{ mod } \mathcal{U}_{d^2+1}(t).
$$
Now we can compute the remainder $s \text{ mod } t$ for given polynomials 
$s,t \in \mathbb{Z}[x_1,\ldots,x_k,y]$ (where the leading coefficient of $t$ is $1$)
in $\DLOGTIME$-uniform $\TC^0$ as follows:
\begin{enumerate}
\item
Compute the number $d  = \max\{ \deg_z(p) \mid  p \in \{s,t\}, z \in \{x_1, \ldots, x_k,y\}\}$, so that
$s,t \in \mathbb{Z}[x_1,\ldots,x_k,y]_d$.
\item 
Compute in parallel $u(z) = \mathcal{U}_{d^2+1}(s)$ and $v(z) = \mathcal{U}_{d^2+1}(t)$.
\item Compute, using Theorem~\ref{lemma-eberly},
$R(z) = u(z) \text{ mod } v(z)$, which is $\mathcal{U}_{d^2+1}(s \text{ mod } t)$.
\item 
Finally, compute $\mathcal{U}^{-1}_{d^2+1}(R)$ which is $s \text{ mod } t$. 
\qed
\end{enumerate}
In the same way we can also compute the quotient, but we only will need the remainder 
$s \text{ mod } t$ in the following.
\end{proof}
Finally, we will also need the following result from \cite{Rob93}:

\begin{theorem}[Theorem~5.2 in \cite{Rob93}] \label{thm-WP-finite-quotient}
Let $G$ be a f.g. group with a normal subgroup $H$ of finite index. Then, the word problem for 
$G$ is $\AC^0$-reducible to the word problems for $H$ and $G/H$.
\end{theorem}
Now we are in the position to prove Theorem~\ref{thm-wp-linear-solvable}.

\medskip

\noindent
{\it Proof of Theorem~\ref{thm-wp-linear-solvable}.}
Let us first assume that $G$ is f.g. solvable and linear over a field $F$.
By a theorem of Mal'cev (see e.g. \cite[Theorem~3.6]{Wehr73}), $G$ contains a normal subgroup $H$ of  finite index, which is triangularizable over a finite extension of $F$.
Using Theorem~\ref{thm-WP-finite-quotient} we know that the word problem for $G$ is $\AC^0$-reducible to the word problems for $H$ and $G/H$. The latter
is a finite solvable group, see Theorem~\ref{thm-nilpotent-torsion-free-subgroup}. Hence, its word problem belongs to $\DLOGTIME$-uniform $\TC^0$
(actually $\mathsf{ACC}^0$) by \cite{BarringtonT88}. 

By the previous discussion, it suffices to show that the word problem for a f.g. triangular matrix group $G$ over some field $F$ belongs to 
$\DLOGTIME$-uniform $\TC^0$. Let $P$ be the prime field of $F$.  
We can replace $F$ by the finitely generated extension of $P$ that is generated by all matrix entries in generators of $G$.
It is known that the field extension $[F:P]$ has a separating transcendence base $\{ x_1, \ldots, x_k\}$, which means that
$[F:P(x_1,\ldots,x_k)]$ is a finite separable extension; see e.g. \cite[Theorem~31]{ZaSa58}.\footnote{Every finitely generated extension field
of a perfect field has a separating  transcendence base and every prime field is perfect.}
Hence, the theorem of the primitive element applies, which says
that $F$ is generated over 
$P(x_1,\ldots,x_k)$ by a single element $\alpha \in F$, which is algebraic over $P(x_1,\ldots,x_k)$.

Assume now that $P = \mathbb{Q}$ (in case $P  = \mathbb{F}_p$ for a prime $p$ we have to replace 
in all arguments below $\mathbb{Z}$ by $\mathbb{F}_p$).
Consider the minimal polynomial $p(y) \in \mathbb{Q}(x_1,\ldots,x_k)[y]$ of $\alpha$. We can write it as
\begin{equation} \label{min-poly}
p(y) = y^m + \frac{p_1}{q} y^{m-1} + \frac{p_2}{q} y^{m-2} \cdots + \frac{p_m}{q} 
\end{equation}
for $p_1, \ldots, p_m, q \in \mathbb{Z}[x_1, \ldots, x_k]$, $q \neq 0$.
The element $\beta = \alpha \cdot q \in F$ also generates $F$ over $P(x_1,\ldots,x_k)$, and 
its minimal polynomial is
$$
q(y) = y^m + p_1 \cdot y^{m-1} +  p_2 q \cdot y^{m-2} + \cdots + p_m q^{m-1} \in \mathbb{Z}[x_1, \ldots, x_k,y]
$$
(multiply \eqref{min-poly} by $q^m$).
We have
$$
F =  \mathbb{Q}(x_1,\ldots,x_k)[y]/\langle q(y) \rangle ,
$$
where $\langle q(y) \rangle = \{ a(x) \cdot q(x) \mid a(x) \in  \mathbb{Q}(x_1,\ldots,x_k)[y] \}$ is the 
ideal generated by $q(x)$.

Each of the finitely many generators of the group $G$ is a matrix, whose entries are polynomials in the variable $y$ with coefficients from the fraction
field $\mathbb{Q}(x_1,\ldots,x_k)$. Every such coefficient is a fraction $a(x_1,\ldots, x_k)/b(x_1,\ldots, x_k)$ with 
$a(x_1,\ldots, x_k)$, $b(x_1,\ldots, x_k) \in  \mathbb{Z}[x_1,\dots,x_k]$. Let $g(x_1,\ldots,x_k)$ be the greatest common divisor 
of all denominators $b(x_1,\ldots, x_k)$, which is a fixed polynomial. 
Instead of asking whether $A_1\cdots A_n \equiv \Id \text{ mod } q(y)$ (for group generators $A_1, \ldots, A_n$ of $G$)
we can ask whether $gA_1 \cdots gA_n \equiv g^n \Id \text{ mod } q(y)$.\footnote{Here, for two $(d \times d)$-matrices $A$ and $B$, $A \equiv B$ mod $q(x)$ 
means that $A[i,j] \equiv B[i,j]$ mod $q(x)$ for all $1 \leq i,j \leq d$.}
So far, the proof has been following more or less closely Waack's arguments from \cite{Waa91}.

Let $M_ i = gA_i$, which is a triangular matrix of dimension $d$ for some fixed $d \in \mathbb{N}$ with entries from $\mathbb{Z}[x_1,\dots,x_k]$. 
Let us write $M_i  = D_i + U_i$, where $D_i$ is a diagonal matrix and $U_i$ is upper triangular with all diagonal entries equal to zero.
We get 
\begin{equation} \label{product-of-M_i}
M_1\cdots M_n=\prod_{i=1}^n(D_i +U_i) = \sum_{X_1 \in \{D_1,U_1\}} \cdots \sum_{X_n \in \{D_n,U_n\}} \prod_{j=1}^n X_j.
\end{equation}
If there are more than $d-1$  factors $U_i$ in a product $\prod_{j=1}^n X_j$, then the product is the zero matrix. 
So there are at most $\sum_{i=0}^{d-1} \binom{n}{i} \leq d \binom{n}{d} \leq d n^d$ summands (for $n>2d$) 
in the sum \eqref{product-of-M_i} that are not equal to zero. When we look at one of the products $\prod_{j=1}^n X_j$ with at most 
$d-1$ many factors $U_i$, we can write it as 
\begin{align*}
& \left( \prod_{i=1}^{m_1-1}D_{i} \right) U_{m_1} \left( \prod_{i=m_{1}+1}^{m_2-1}D_{i}\right)  \cdots  U_{m_l} \left(\prod_{i=m_l+1}^{n}D_{i} \right)
= \\
& D_{1,m_1-1} U_{m_1}  D_{m_1+1,m_2-1}  \cdots U_{m_l} D_{m_l+1,n}
\end{align*}
for some $0 \leq l \leq d-1$ and $1 \leq m_1 < \cdots < m_l \leq n$,
where $D_{u,v} = \prod_{i=u}^v D_i$ ($1 \leq u \leq v+1$, $0 \leq v \leq n$) is a product of at most $n$ diagonal matrices.
Each of these products can be calculated by calculating $d$ products of at most $n$ polynomials
from $\mathbb{Z}[x_1,\dots,x_k]$, which can be done  in $\DLOGTIME$-uniform $\TC^0$ by Lemma~\ref{lemma-poly-mult-TC0}. 
Moreover, all products $D_{u,v}$ for $1 \leq u < v \leq n$ can be computed in parallel.
Once these products are computed, we can, in parallel, compute for 
all $0 \leq l \leq d-1$ and $1 \leq m_1 < \cdots < m_l \leq n$ the matrix product
$D_{1,m_1-1} U_{m_1}  D_{m_1+1,m_2-1}  \cdots U_{m_l} D_{m_l+1,n}$. Note that these products
have constant length and hence involve a constant number of polynomial multiplications and additions.
So, all the above matrix products can be computed in $\DLOGTIME$-uniform $\TC^0$ as well. Next, we have to compute 
the sum of all polynomially many matrices computed in the previous step. For this we have to compute
$d^2$ many sums of polynomially many polynomials, which is again possible in $\DLOGTIME$-uniform $\TC^0$.
The resulting matrix is  $M_1\cdots M_n = g^n A_1 \cdots A_n$.
Finally we have to reduce all entries of the matrices $M_1 \cdots M_n$ and $g^n \Id$ modulo the minimal polynomial $q(y)$ 
which can also be done in $\DLOGTIME$-uniform $\TC^0$ by Lemma~\ref{lemma-poly-div-TC0}. Note that we 
divide by the polynomial $q(y)$, whose leading coefficient is indeed $1$.

Finally, let $G$ be a f.g. virtually solvable linear group $G$. Then $G$ contains a normal solvable subgroup $H$, for which we know that
the word problem can be solved in $\DLOGTIME$-uniform $\TC^0$. Moreover, the quotient $G/H$ is a finite group, for which the word problem belongs to 
$\DLOGTIME$-uniform $\NC^1$. Hence,  Theorem~\ref{thm-WP-finite-quotient} implies that the word problem for $G$ belongs to $\DLOGTIME$-uniform $\NC^1$.
\qed

\end{document}